\documentclass[journal,12pt,onecolumn,draftclsnofoot]{IEEEtran}
\usepackage{wrapfig}
\usepackage{amsfonts}
\usepackage{amssymb,epic,eepic}
\usepackage{amsmath}
\usepackage{dsfont}
\usepackage{dcolumn}
\usepackage{bm}
\usepackage{bbm}
\usepackage{verbatim}
\usepackage{color}
\usepackage{amsthm}
\usepackage{tikz}
\usetikzlibrary{positioning}
\usepackage{MnSymbol,wasysym}
\usetikzlibrary{arrows,positioning,decorations.pathmorphing,decorations.markings,shapes,fadings}
\usepackage{pstool}
\usepackage{psfrag}
\usepackage{import}
\usepackage{varwidth}
\usepackage{nicefrac}
\newcommand{\nn}{{\mathbb N}}

\newcommand{\eps}{{\varepsilon}}        

\newcommand{\cP}{\mathcal P}

\newcommand{\cC}{\mathcal C}

\newtheorem{theorem}{Theorem}

\newtheorem*{acknowledgement}{Acknowledgement}

\newtheorem{definition}{Definition}
\newtheorem{example}{Example}

\newtheorem{remark}{Remark}

\newcommand{\tr}{\mathrm{tr}}

\makeatletter
\newcommand*{\rom}[1]{\expandafter\@slowromancap\romannumeral #1@}
\makeatother
\DeclareMathOperator{\conv}{conv}

\DeclareMathOperator{\linspan}{span}

\DeclareRobustCommand*{\IEEEauthorrefmark}[1]{\raisebox{0pt}[0pt][0pt]{\textsuperscript{\footnotesize #1}}}
\tikzset{
	schraffiert/.style={pattern=horizontal lines,pattern color=#1},
	schraffiert/.default=black
}
\usetikzlibrary{patterns}

\begin{document}
	\title{Reliable Communication under the Influence of a State-Constrained Jammer: A Novel Perspective on Receive Diversity}
	\author{Christian Arendt\IEEEauthorrefmark{1,2}, Janis N\"otzel\IEEEauthorrefmark{3,4} and Holger Boche\IEEEauthorrefmark{2}\\
		\scriptsize
		\vspace{0.2cm}                               
		\IEEEauthorblockA{\IEEEauthorrefmark{1}
		BMW Group, 80788 M\"unchen, Germany, Email: christian.ca.arendt@bmw.de}
	
	\IEEEauthorblockA{\IEEEauthorrefmark{2}
		Lehrstuhl f\"ur Theoretische Informationstechnik, Technische Universit\"at M\"unchen, 80290 M\"unchen, Germany, Email: boche@tum.de}
	\IEEEauthorblockA{\IEEEauthorrefmark{3}
		Theoretische Nachrichtentechnik, Technische Universit\"at Dresden, 01187 Dresden, Germany, Email: janis.noetzel@tu-dresden.de}\\
	\IEEEauthorblockA{\IEEEauthorrefmark{4}
		F\'{\i}sica Te\`{o}rica: Informaci\'{o} i Fen\`{o}mens Qu\`{a}ntics, Universitat Aut\`{o}noma de Barcelona, 08193 Bellaterra, Spain}
}
	\maketitle
\begin{abstract}
The question of robust direct communication in vehicular networks is discussed. In most state-of-the-art approaches, there is no central entity controlling channel access, so there may be arbitrary interference from other parties. Thus, a suitable channel model for Vehicle-to-X (V2X) communication is the Arbitrarily Varying Channel (AVC). Employing multiple antennas on a vehicle or sending over multiple frequencies to make use of diversity are promising approaches to combat interference. In this setup, an important question about diversity is how many antennas or orthogonal carrier frequencies are necessary in order to avoid system breakdown due to unknown interference in AVCs. For Binary Symmetric AVCs (AVBSC) and a physically meaningful identical state-constrained jammer, the deployment of a third, uncorrelated receiving antenna or the parallel transmission over three different orthogonal frequencies avoids symmetrizability and thus ensures positivity of the capacity of the overall communication channel. Furthermore, the capacity of the identical state-constrained composite AVBSC is continuous and shows super-activation, a phenomenon which was hitherto deemed impossible for classical communication without secrecy constraints. Subsuming, spatial and frequency diversity are enablers for reliable communication over communication channels with arbitrarily varying interference.
\end{abstract}
{\smallskip\bf Keywords:} vehicular connectivity, unknown interference, receive diversity, arbitrarily varying channel, super-activation.
\begin{section}{Introduction}
A prominent technology for Vehicle-to-X (V2X) communication in vehicular ad-hoc networks is an extension of the wireless LAN standard IEEE 802.11a, called IEEE 802.11p (11p) \cite{5514475}. Based on this standard, Dedicated Short Range Communication (DSRC) was developed in the United States \cite{5888501}. In order to avoid confusion, the term ITS-G5 \cite{its2011g5} is used for DSRC in Europe. Since 11p is intended for direct short-range and low-latency communication, there is no central entity controlling the use of spectral resources. For this reason, the frequency band extending from 5.85 GHz to 5.925 GHz is used in a shared, self-coordinated manner. Carrier sense multiple access/collision avoidance is applied as a medium access scheme to avoid packet collisions. Though there exist power control and adaption schemes, these techniques cannot prevent simultaneous channel access where the interference produced by the respective entities is generally unknown. This interference can also be caused by coexistence of different technologies in the same frequency band like, for example, 11p and Device-to-Device (D2D) communication in Long Term Evolution (LTE) for vehicle-to-vehicle communication (designated as PC5) \cite{7859501}. In this setup, packet collisions can happen causing unbounded delays in actual communication schemes, for example, in 11p dropping collided packets---a serious drawback for time-critical information exchange \cite{7313128}. Additionally, in scenarios with fast and arbitrarily varying channel conditions due to blocking or the influence of secondary moving objects in the surrounding of a vehicle, the channel conditions may not be tracked accurately by channel estimation techniques. Thus, the receiver is completely unaware of the actual fading conditions. In order to be able to enable new use cases in the context of vehicular connectivity with high reliability requirements, the communication hardware has to be designed to work properly even under worst-case conditions.\\
A well-studied model to account for the uncertainty in channel noise/interference is the Arbitrarily Varying Channel (AVC) introduced by Blackwell et al. in \cite{blackwell1960}. In addition, AVCs and the communication under channel uncertainty are generally discussed in \cite{720535} and \cite{csiszár2011information}. The AVC is based on the principle that a jammer can control the state of the channel in an arbitrary manner.\\
One technique to combat interference, thereby increasing the quality and reliability of a wireless link, is the deployment of multiple antennas at the transmitter and/or receiver, resulting in multiple-input-single-output, Single-Input-Multiple-Output (SIMO) or Multiple-Input-Multiple-Output (MIMO) schemes, respectively. Besides increasing data rate, these systems exploit antenna diversity, that is, every antenna observes a more or less uncorrelated version of the transmitted signal(s). However, to the author's knowledge, there exists no theory relating diversity techniques to the symmetrizability of AVCs. This gap is closed by this contribution showing the potential of different diversity schemes in scenarios with arbitrarily varying interference.\\
An important question, which has to be answered in the context of safety related information exchange, is whether it is possible to reliably communicate over AVCs, that is, using coding techniques to correct errors caused by collisions. Communication over AVCs at positive rates can be possible using Common Randomness (CR)-assisted coding, even when it is impossible without \cite{AHLSWEDE1969457}. CR-assisted coding is characterized by a transmitter (Tx) and a receiver (Rx) observing perfectly correlated outcomes of a random experiment. CR can be established by a satellite signal or a common synchronization procedure \cite{7217803}. It is used to coordinate the choice of a specific code out of a common codebook library. A CR-assisted code is explicitly defined in Definition~\ref{def:CR_ass_code}. In contrast, for deterministic coding, communication at positive rates over AVCs cannot be guaranteed, since the AVC can be symmetrizable. In this situation, the jammer may choose its inputs such that any two codewords may be confused at the decoder. In \cite{2627}, Csiszar and Narayan deduced that non-symmetrizability of an AVC is a sufficient condition for the positivity of the deterministic code capacity under average error criterion. In addition, a famous result by Ahlswede \cite{Ahlswede1978} states that the deterministic code capacity of a discrete memoryless AVC either equals its CR-assisted code capacity or else is zero (Ahlswede dichotomy). Thus, common randomness shared between transmitter and receiver can boost the capacity away from zero. Explicit examples for this effect are known, see for example, \cite{Ahlswede1978}. In the previously mentioned contribution, an example of a symmetrizable point-to-point channel that has zero capacity for deterministic codes, but which has strictly positive capacity for CR-assisted codes, is given. Transferring the concept of CR-assisted coding to a practical application, the source of common randomness is replaced by the concept of correlated sources. This technique is intensively investigated in \cite{605589,6883596} and \cite{6620445}. In \cite{6883596}, the authors show that for an arbitrarily varying broadcast channel no more than $\mathcal{O}(\log n)$ outputs of correlated sources at block length $n$ are sufficient to achieve the same capacity as by using coordination established by common randomness. Nevertheless, in order to guarantee reliable message transmission under the influence of a jammer when the channel is symmetrizable, all previous work assumes an external assistance mechanism that cannot be accessed by the jammer---be it a side channel or source. Summarizing, using CR-assisted coding is a promising approach to combat arbitrarily varying interference, in theory. Nevertheless, when it comes to practical implementation, the availability of a common source of randomness shared by encoder and decoder is equivalent to the use of an additional control channel, as used for example in mode 3 in LTE-D2D \cite{7859501}, which may not always be available. Thus, the practical relevance of deterministic codes for reliable safety-related communication in vehicular networks is indisputable.\\
In classical communication, transmitting over multiple useless channels suggests zero capacity of the overall system. In contrast, the phenomenon of super-activation, which is well known from quantum physics, states that the classical additivity arguments of basic resources are not valid in general as first observed by Smith and Yard in \cite{smith2008quantum}. For classical channel models, super-activation, which is the strongest form of violation of additivity, has first been demonstrated in the context of arbitrarily varying wiretap channels, see for example, \cite{6620445,6875259,7447794,7541865}. Wiretap channels account for secrecy, that is, information protection against eavesdropping in communications in an information-theoretic setting. For this class of channels, the joint use of multiple orthogonal channels enables super-activation for the secrecy capacity, the first appearance in a classical setting \cite{6620445}, \cite{7447794,7541865}. Super-additivity is a weaker form of violation of the additivity of the capacity of orthogonal channels. The phenomenon of super-additivity has a long history in classical information theory: Ahlswede showed in \cite{ahlswede1970} that the treatment of capacity expressions in the context of AVCs under maximum error and deterministic coding is strongly related to Shannon's zero-error capacity problem for a Discrete Memoryless Channel (DMC). Shannon suggested in \cite{1056798} in 1956 that the zero-error capacity of orthogonal DMC is additive. Forty years later, Alon \cite{Alon1998} disproved this conjecture by demonstrating super-additivity of the zero-error capacity. For orthogonal AVCs, the message transmission capacity under average error criterion cannot be super-activated. Instead, super-additivity may occur \cite{7541865}.\\
In this contribution, we present an interference combating concept based on different diversity schemes and a state-constrained jammer. In contrast to \cite{2627}, our scheme focuses on avoiding symmetrizability rather than enabling communication at positive rates over symmetrizable channels by imposing power constraints, although we also explore the impact of latter in Theorem~\ref{thm:super_activ_sc_pw}. In addition, we differentiate from \cite{1056995} since our findings, with one exception, are based on deterministic coding schemes. We discuss channel models which exemplary account for the following real world scenarios:
\begin{enumerate}
	\item MIMO communication with a state-constrained jammer (one input controlling the channel state of all diversity branches simultaneously) due to a limited number of transmitting antennas at the jammer or a trivial precoding strategy,
	\item Frequency diversity interfered by a second, uncoordinated Tx-Rx pair operating in trivial diversity mode (same input symbol sent over all operating frequencies),
	\item Trivial spatial receive diversity (same input symbol sent over all parallel channels resulting in one 'effective' channel) to compensate arbitrary interference caused by a state-constrained jammer (see 1.),
	\item Trivial Frequency diversity (same input symbol sent over all operating frequencies) interfered by a second, uncoordinated transmitter-receiver pair using trivial diversity.
\end{enumerate}
	In the scenarios relying on spatial diversity, the jammer can also be replaced by a second, uncoordinated Tx-Rx pair operating in the same frequency band. The uncoordinated Tx-Rx pair can be located inside the vehicle, for example, when using multiple uncoordinated connectivity modules simultaneously operating in the same frequency band, or can be embodied by two additional communicating parties. A visualization for the previously introduced AVC scenarios in the context of vehicular communication is shown in Figure~\ref{fig:car_2_car_jamming}.
\begin{figure}
	\setlength\abovecaptionskip{-8pt}
	\centering
	\includegraphics*[width=1\textwidth]{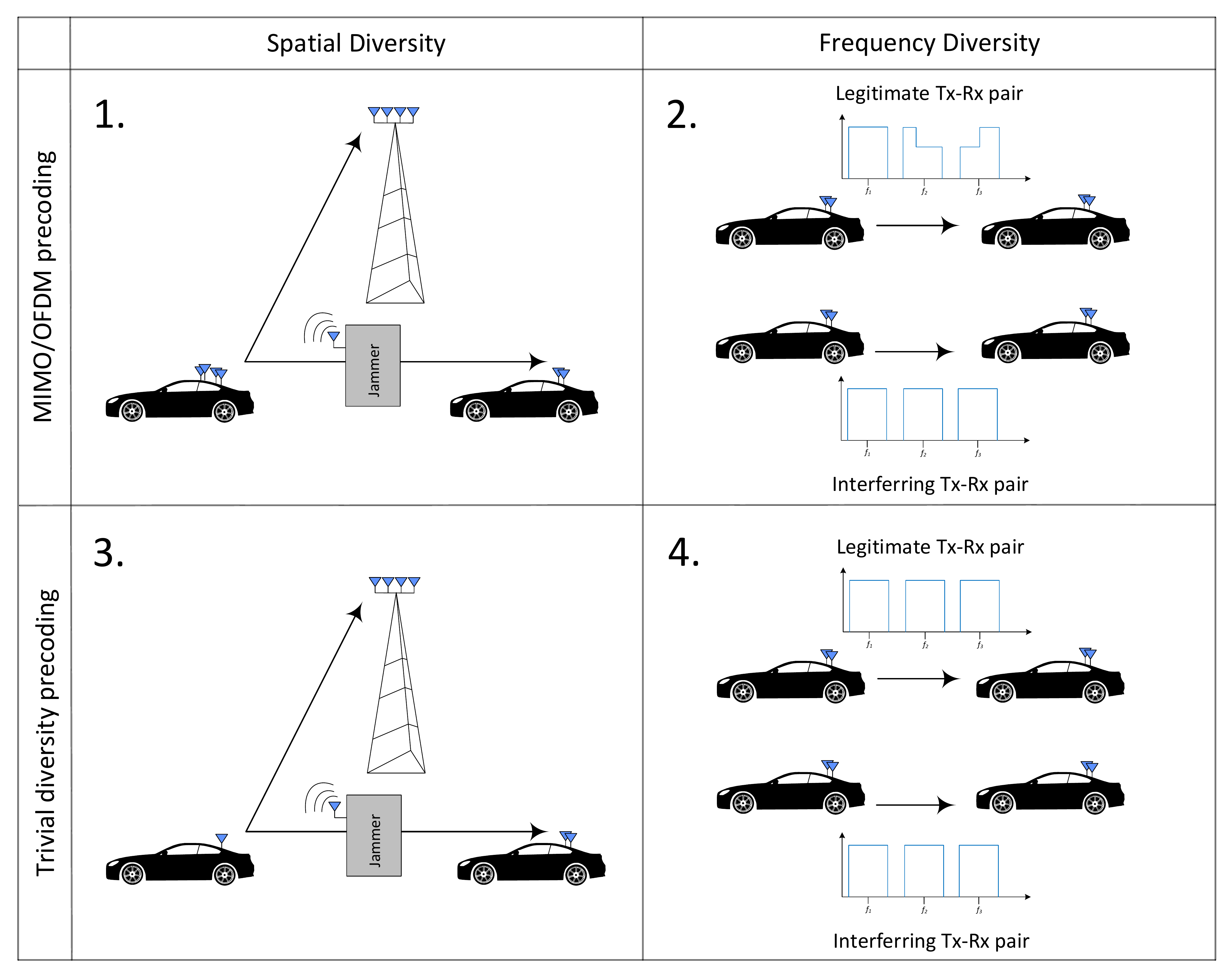}
	\vspace{-0.5 cm}
	\caption{AVC scenario in vehicular communication for varying diversity schemes: The channel state is controlled by an adversary (jammer) which causes arbitrary, unknown interference or, alternatively, interference is caused by a second, uncoordinated Tx-Rx pair. Reliable communication has to be ensured for any strategy of the jammer or any interference caused by the secondary communication parties.}
	\label{fig:car_2_car_jamming}
\end{figure}
	We will link the channel models, introduced in Section~\ref{sec:system model}, to the previously introduced scenarios at the appropriate passages.\\
	One of the main objectives of this contribution is to show that in the case of the identical state-constrained (i.s.c.) composite independent AVBSC, super-activation occurs. Thus, the capacity of such a composite system can be positive and continuous, even in situations where each individual capacity is zero and discontinuous as a function of the system parameters. Continuity, which is well-known for CR-assisted coding as discussed in \cite{wiese-noetzel-boche-I}, is not valid for deterministic coding for AVCs in general \cite{doi:10.1063/1.4902930}. For the AVC, the literature shows that discontinuity is a generic feature \cite{wiese-noetzel-boche-I,doi:10.1063/1.4902930}. The overarching aim of our investigations is to enlarge the capacity of the overall communication channel as well as to prevent complete system breakdown caused by symmetrizability by using diversity schemes.
	\linebreak\\
		{\bf Outline of the paper.} The remainder of the paper is structured as follows: First, we discuss the system model under consideration, including implications for practical applications, and introduce basic channel models in Section~\ref{sec:system model}. In Section~\ref{sec:code_concepts_and_perf_measures}, code concepts and performance measures are established. Next, we state our main result in Section~\ref{sec:main_result}. This shows that except in trivial cases, the adjoining of a third, uncorrelated orthogonal channel already avoids symmetrizability for an i.s.c. composite independent AVBSC. In addition, we show that the capacity of an i.s.c. composite AVBSC shows super-activation if respective power constraints are imposed on the jammer and the transmitter, or if channel state combinations of little practical relevance are excluded. Next, we transfer our results to the i.s.c. composite orthogonal AVBSC. The proof of our main theorem together with the proof of further results can be found in Section~\ref{sec:proofs}. Finally, Section~\ref{sec:discussion_and_conclusion} discusses the implications of the derived results and gives hints for further investigations.
	\end{section}	
\begin{section}{System Model and Basic Channel Models}\label{sec:system model}
We adapt our notation to the one presented in \cite{7447794}, \cite{7572096} and \cite{noetzel2016}: We set $x':=1-x$ for any number $x\in\mathbb R$. For $L\in\nn$, we define $[L]:=\{1,\ldots,L\}$. We denote the set of permutations on $[L]$ by $S_L$. Let two sets $\mathcal{X}$ and $\mathcal{Y}$ of cardinality $|\mathcal X|=L_\mathcal{X}$ and $|\mathcal Y|=L_\mathcal{Y}$ with $L_\mathcal{X},L_\mathcal{Y}\in\nn$ be given. Their product can be calculated in the following way: 
\begin{align}
\mathcal{X}\times\mathcal{Y}:=\{(x,y):x\in\mathcal{X},\ y\in\mathcal{Y}\}. 
\end{align}
Additionally, $\mathcal{X}^n$ is the n-fold product of $\mathcal{X}$ with itself for any $n\in\nn$. The set of probability distributions on a finite set $\mathcal{X}$ is denoted by
\begin{align}
\cP(\mathcal{X}):=\{p:\mathcal{X}\to\mathbb R\ :\ p(x)\geq0\ \forall\ x\in\mathcal{X},\ \sum_{x\in\mathcal{X}}p(x)=1\}.
\end{align}
An important subset of elements of $\cP(\mathcal{X})$ is the set of its extremal points, the Dirac-measures: For $x_1,x_2\in\mathcal{X}$, $\delta_{x_1}\in\cP(\mathcal{X})$ is defined through $\delta_{x_1}(x_2)=\delta(x_1,x_2)$, where $\delta(\cdot,\cdot)$ is the usual Kronecker-delta symbol. In order to exploit structural properties, we transfer the probabilistic concepts to linear algebra by considering $\cP(\mathcal{X})$ as being embedded into $\mathbb R^{L_\mathcal{X}}$ through the bijection $p\mapsto \sum_{x\in\mathcal{X}}p(x)e_x$.	Under this transformation, $\delta_x$ is mapped to the fixed basis $\{e_x\}_{x\in\mathcal X}$. This allows a natural use of matrix calculus in our analysis. We solely introduce results from multi-linear algebra for bipartite systems. The generalization to the multi-partite case is straightforward. Besides $\{e_x\}_{x\in\mathcal X}$ for $\mathbb R^{L_\mathcal{X}}$, we use a second fixed basis $\{e_y\}_{y\in\mathcal Y}$ for $\mathbb R^{L_\mathcal{Y}}$. $L_\mathcal{X}\times L_\mathcal{Y}$ matrices define linear maps from $\mathbb R^{L_\mathcal{X}}$ to $\mathbb R^{L_\mathcal{Y}}$ via their actions in these bases. The tensor product of $\mathbb R^{L_\mathcal{X}}$ with $\mathbb R^{L_\mathcal{Y}}$ is $\mathbb R^{L_\mathcal{X}}\otimes\mathbb R^{L_\mathcal{Y}}:=\linspan\{e_x\otimes e_{y}\}_{x,\in\mathcal{X},y\in\mathcal{Y}}$. This allows us to define general so called 'product vectors' of two vectors $u=\sum_{x\in\mathcal{X}}u_xe_x$ and $v=\sum_{y\in\mathcal{Y}}v_ye_y$ by
\begin{align}
u\otimes v:=\sum_{x\in\mathcal{X},y\in\mathcal{Y}}u_xv_{y}e_x\otimes e_{y}.
\end{align}
The vector space $\mathbb R^L\otimes\mathbb R^{L'}$ inherits the scalar product by the formula $\langle u\otimes v,x\otimes y\rangle:=\langle u,x\rangle\langle v,y\rangle$. The space of $L\times L'$ matrices is denoted by $M_{L\times L'}$. Given $A,B\in M_{L\times L'}$, we define $A\otimes B$ through its action on product vectors: 
\begin{align}
(A\otimes B)(u\otimes v):=(Au)\otimes(Bv).
\end{align}
In order to simplify notation later, for $u\in\mathbb R^{L_\mathcal{X}}$ and $n\in\nn$, we will use the shorthand $u^{\otimes n}:=u\otimes\ldots\otimes u$ for the $n$- fold tensor product of $u$ with itself. Accordingly, for $A\in M_{L_\mathcal{X}\times L_\mathcal{Y}}$, we write $A^{\otimes n}:=A\otimes\ldots\otimes A$.\\
\indent The influence of noise during the transmission of messages is modeled by stochastic matrices $W$ of conditional probability distributions $\left(w\left(y\vert x\right)\right)_{x\in\mathcal{X},y\in\mathcal{Y}}$, whose entries satisfy $\forall x\in\mathcal{X}:w\left(\cdot\vert x\right)\in\cP(\mathcal{Y})$. Any such matrix is henceforth also called a channel. The set of channels acting on a finite alphabet $\mathcal{X}$ of size $L_\mathcal{X}$ and $\mathcal{Y}$ of size $L_\mathcal{Y}$ is denoted by $\cC(\mathcal{X},\mathcal{Y})$. The special case where $\forall x\in\mathcal{X}$, $y\in\mathcal{Y}:\ w(y|x)=\delta(y,x)$ is denoted by $Id$. An important sub-class of channels arises for a binary alphabet ($L_\mathcal{X}=L_\mathcal{Y}=2$). In this case, every channel matrix $W$ is completely characterized by two parameters in the following sense:
\begin{align}
	W&=\begin{pmatrix}
	w_1 &w_2\\
	1-w_1 &1-w_2
	\end{pmatrix}.
\end{align}
If $w_2=1-w_1$, the resulting channel is called a 'Binary Symmetric Channel' (BSC). Any BSC with crossover probability $1-w$, for which we use the shorthand $BSC(w)$, is completely defined via the BSC-parameter $w$. A special case of a BSC is the bit-flip channel where $w=0$ and which is denoted by $\mathbb F$. In order to accurately model the influence of specific jamming strategies in a probabilistic framework, we introduce the notion of an Arbitrarily Varying Channel (AVC): The probabilistic law governing the transmission of codewords over a point-to-point AVC for $n$ channel uses is described by
\begin{align}\label{eq:AVC_general}
w^{\otimes n}(y^n\vert x^n,s^n)=\prod\limits_{i=1}^nw(y_i\vert x_i,s_i),
\end{align}
where $s^n=(s_1,\ldots,s_n)\in\mathcal{S}^n$ are the jammer's or adversarial's inputs controlling the state of the channel, $x^n=(x_1,\ldots,x_n)\in\mathcal{X}^n$ are the input codewords of the encoder and $y^n=(y_1,\ldots,y_n)\in\mathcal{Y}^n$ denotes the channel outputs at the decoder, all assumed to be taken from finite alphabets. The previously introduced notion naturally extends to products of AVCs. Let, for example, $K=2$ DMCs denoted by $W_1$ and $W_2$ mapping $\mathcal{X}_1$ to $\mathcal{Y}_1$ and $\mathcal{X}_2$ to $\mathcal{Y}_2$ respectively, both with transition probability matrices $(w_1(y_1\vert x_1))_{x_1\in\mathcal{X}_1,y_1\in\mathcal{Y}_1}$ and $(w_2(y_2\vert x_2))_{x_2\in\mathcal{X}_2,y_2\in\mathcal{Y}_2}$. Then, the transition probability matrix of $W_1\otimes W_2$ is defined by $w(y_1,y_2\vert x_1,x_2):=w_1(y_1\vert x_1)\cdot w_2(y_2\vert x_2)$, for all $x_1\in\mathcal{X}_1, x_2\in\mathcal{X}_2, y_1\in\mathcal{Y}_1, y_2\in\mathcal{Y}_2$. This notation can be adapted to arbitrarily varying channels:
	\begin{definition}[AVC]\label{def:AVC_single}	
	Let $\mathcal S,\mathcal X,\mathcal Y$ be finite sets. An arbitrarily varying channel (AVC) is a collection $\mathcal W:=(W_s)_{s\in\mathcal S}$ for which $W_s\in\mathcal C(\mathcal X,\mathcal Y)$ for every $s\in\mathcal S$. Alternative useful ways of writing the channel are $w(y|x,s):=w_s(y|x)$ or $W_s(\delta_x)=w_s(\cdot|x)$ for all $s\in\mathcal S$, $x\in\mathcal X$ and $y\in\mathcal Y$. The action of the AVC is completely described by the sequence $((W_{s^n})_{s^n\in\mathcal S^n})_{n\in\mathbb{N}}$ where $W_{s^n}=W_{s_1}\otimes\ldots\otimes W_{s_n}$.		
\end{definition}
	\definecolor{light-gray}{gray}{0.8}
	\definecolor{dark-gray}{gray}{0.4}
	\definecolor{extremelight-gray}{gray}{0.9}
	\definecolor{darkblue}{RGB}{0,0,90}
	\definecolor{lightblue}{RGB}{102,153,255}
	\setlength\abovecaptionskip{-1pt}
	\begin{figure}\centering\begin{tikzpicture}[thick,scale=0.85, every node/.style={transform shape}]
		\node (T2)  [fill=light-gray,rectangle, minimum width=2cm,minimum height=4.5cm] at (5,1.5) {\textcolor{black}{Central Unit}};
		\node (T3)  [fill=darkblue,rectangle,inner sep=6pt] at (1,0.0) {\textcolor{white}{MRF \#N}};
		\node (T4)  [fill=darkblue,rectangle,inner sep=6pt] at (1,2) {\textcolor{white}{MRF \#2}};
		\node (T5)  [fill=darkblue,rectangle,inner sep=6pt] at (1,3.0) {\textcolor{white}{MRF \#1}};
		\node (A1) [black, ellipse,rotate=90,inner sep=4.5pt] at (3,1) {{\Large ...}};
		\path[color=dark-gray,line width=1pt] (3.7,3) edge (2,3);
		\path[color=dark-gray,line width=1pt] (3.7,2) edge (2,2);
		\path[color=dark-gray,line width=1pt] (3.7,0) edge (2,0);
		\path[color=black,line width=1pt] (0,3) edge (-0.2,3);
		\path[color=black,line width=1pt] (-0.2,3) edge (-0.2,3.4);
		\path[color=black,line width=1pt] (-0.2,3.4) edge (-0.4,3.6);
		\path[color=black,line width=1pt] (-0.2,3.4) edge (0,3.6);
		\path[color=black,line width=1pt] (0,3.6) edge (-0.4,3.6);
		\path[color=black,line width=1pt] (0,2) edge (-0.2,2);
		\path[color=black,line width=1pt] (-0.2,2) edge (-0.2,2.4);
		\path[color=black,line width=1pt] (-0.2,2.4) edge (-0.4,2.6);
		\path[color=black,line width=1pt] (-0.2,2.4) edge (0,2.6);
		\path[color=black,line width=1pt] (0,2.6) edge (-0.4,2.6);
		\path[color=black,line width=1pt] (0,0) edge (-0.2,0);
		\path[color=black,line width=1pt] (-0.2,0) edge (-0.2,0.4);
		\path[color=black,line width=1pt] (-0.2,0.4) edge (-0.4,0.6);
		\path[color=black,line width=1pt] (-0.2,0.4) edge (0,0.6);
		\path[color=black,line width=1pt] (0,0.6) edge (-0.4,0.6);
		\node (T6)  [fill=light-gray,rectangle, minimum width=2cm,minimum height=4.5cm] at (-4,1.5) {\textcolor{black}{Transmitter}};
		\path[color=black,line width=1pt] (-2.8,0) edge (-2.5,0);
		\path[color=black,line width=1pt] (-2.5,0) edge (-2.5,0.4);
		\path[color=black,line width=1pt] (-2.5,0.4) edge (-2.7,0.6);
		\path[color=black,line width=1pt] (-2.5,0.4) edge (-2.3,0.6);
		\path[color=black,line width=1pt] (-2.7,0.6) edge (-2.3,0.6);
		\path[color=black,line width=1pt] (-2.8,2.0) edge (-2.5,2.0);
		\path[color=black,line width=1pt] (-2.5,2.0) edge (-2.5,2.4);				\path[color=black,line width=1pt] (-2.5,2.4) edge (-2.7,2.6);
		\path[color=black,line width=1pt] (-2.5,2.4) edge (-2.3,2.6);
		\path[color=black,line width=1pt] (-2.7,2.6) edge (-2.3,2.6);
		\path[color=black,line width=1pt] (-2.8,3) edge (-2.5,3);
		\path[color=black,line width=1pt] (-2.5,3) edge (-2.5,3.4);				\path[color=black,line width=1pt] (-2.5,3.4) edge (-2.7,3.6);
		\path[color=black,line width=1pt] (-2.5,3.4) edge (-2.3,3.6);
		\path[color=black,line width=1pt] (-2.7,3.6) edge (-2.3,3.6);
		\node (A1) [black, ellipse,rotate=90,inner sep=4.5pt] at (-2.5,1.3) {{\Large ...}};
		\path[color=dark-gray,line width=0.1pt, shorten <= 0pt] (-2.2,0.4) edge (-0.5,3.4);
		\path[color=dark-gray,line width=0.1pt, shorten <= 0pt] (-2.2,2.4) edge (-0.5,2.4);
		\path[color=dark-gray,line width=0.1pt, shorten <= 0pt] (-2.2,3.4) edge (-0.5,0.4);
		\path[color=dark-gray,line width=0.1pt, shorten <= 0pt] (-2.2,3.4) edge (-0.5,3.4);
		\path[color=dark-gray,line width=0.1pt, shorten <= 0pt] (-2.2,2.4) edge (-0.5,0.4);
		\path[color=dark-gray,line width=0.1pt, shorten <= 0pt] (-2.2,0.4) edge (-0.5,0.4);
		\path[color=dark-gray,line width=0.1pt, shorten <= 0pt] (-2.2,3.4) edge (-0.5,2.4);
		\path[color=dark-gray,line width=0.1pt, shorten <= 0pt] (-2.2,0.4) edge (-0.5,2.4);
		\path[color=dark-gray,line width=0.1pt, shorten <= 0pt] (-2.2,2.4) edge (-0.5,3.4);
		\node (T6)  [fill=lightblue,rectangle, minimum width=2cm,minimum height=1cm] at (-1,-1.5) {\textcolor{black}{Jammer}};
		\path[color=black,line width=1pt] (-1,-1) edge (-1,-0.5);
		\path[color=black,line width=1pt] (-1,-0.5) edge (-1.2,-0.3);
		\path[color=black,line width=1pt] (-1,-0.5) edge (-0.8,-0.3);
		\path[color=black,line width=1pt] (-1.2,-0.3) edge (-0.8,-0.3);
		\draw[bend right=30] (-0.7,-0.2) to (-1.3, -0.2);
		\draw[bend right=30] (-0.5,-0.05) to (-1.5, -0.05);
		\end{tikzpicture}\
		\caption{System Model of a communication link involving a Distributed Antenna System (DAS) on the receiver's end. Analog-to-digital conversion is done at the feeding point of the antenna at the Mobile Radio Front-ends (MRF). The digital data of multiple MRFs is merged by the central unit.}
		\label{fig:DAS_system_concept}
	\end{figure}
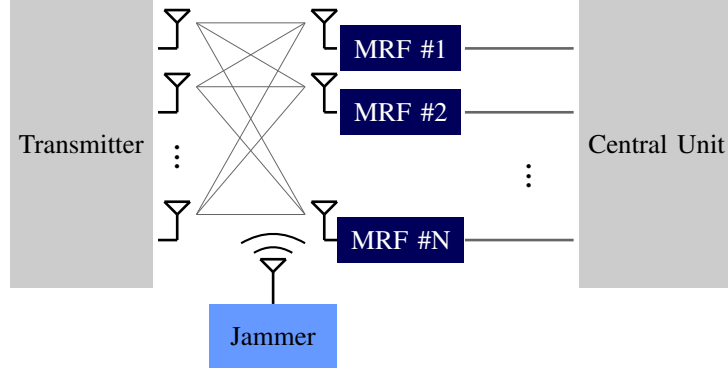
	In later analysis, we make use of the Shannon entropy of $p\in\mathcal{P}(\mathcal{X})$ which is defined as $H(p):=-\sum_{x\in\mathcal{X}}^{}p(x)\log(p(x))$. Every channel $W:\mathcal{P}(\mathcal{X})\mapsto\mathcal{P}(\mathcal{Y})$ together with a probability distribution $p\in\mathcal{P}(X)$ defines a joint distribution	$\mathcal{P}((X,Y)=(x,y))=p(x)w(y\vert x)$ for all $x\in\mathcal{X}$ and $y\in\mathcal{Y}$. Using the previously introduced notation, the mutual information which, by default, is defined as $I(X;Y)=H(X)-H(X\vert Y)$, can equivalently be written as $I(p;W):=I(X,Y)$.\\
	We concentrate on the communication over multiple, orthogonal channels using spatial or frequency diversity. In our scenario, the orthogonality assumption for spatial diversity is justified by the fact that multiple antennas are located in one and the same vehicle in a distributed manner resulting in a so-called Distributed Antenna System (DAS). The concept of DAS includes a central control unit which is connected to the Mobile Radio Front-ends (MRF) via Ethernet or optical fiber. A systematic overview of a DAS is provided in Figure~\ref{fig:DAS_system_concept}. The advantage of a DAS over a commonly used co-located antenna system is the low correlation between the single antenna elements resulting in 'approximately' independent communication channels. To ensure orthogonality in frequency, we assume that no inter-carrier interference occurs in the applied orthogonal frequency-division multiplexing scheme. To model the communication over multiple AVCs, we extend the definition of a single AVC to a composite AVC accounting for the joint use of multiple AVCs in parallel. In our analysis, we distinguish between two different types of composite AVCs: the composite orthogonal AVC and the composite independent AVC according to the respective precoding strategy. The composite orthogonal AVC is characterized by the joint use of multiple orthogonal AVCs allowing for input tuples of the form $x=(x_1,\ldots,x_K)$. In contrast, the composite independent AVC is composed of $K$ independent AVCs all with the same input symbol, that is, setting $x_1=\ldots x_K=x$. It can be deduced that the composite independent AVC is a special case of the composite orthogonal AVC.\\ 
	\begin{definition}[Composite orthogonal AVC]\label{def:Composite_orthogonal_AVC}
	Let the AVCs $\mathcal{W}_i$ with $i\in[K]$ with state set $\mathcal{S}_1,\ldots,\mathcal{S}_K$ be given. We define the composite orthogonal AVC by
	$\mathfrak W:=\otimes_{i=1}^K \mathcal W_i$. It holds
	\begin{align}
	\mathfrak W(y_1,\ldots,y_K\vert x_1,\ldots,x_K,s_1,\ldots,s_K):=\prod\limits_{i=1}^Kw_i(y_i\vert x_i,s_i).
	\end{align}
	\end{definition}
	In order to simplify our analysis, we consider a diversity scheme based on trivial precoding, that is, the same input symbol $x$ is transmitted over all AVCs simultaneously. This implies that $\mathcal{X}_1\times\ldots\times\mathcal{X}_K=\mathcal{X}$ and defines a composite AVC arising from the joint use of $K$ independent AVCs.
	\begin{definition}[Composite independent AVC]\label{def:Composite_independent_AVC}
	Let the AVCs $\mathcal{W}_i$ with $i\in[K]$ with state set $\mathcal{S}_1,\ldots,\mathcal{S}_K$ be given. We define the composite independent AVC by $\mathfrak W_\mathrm{CI}=\mathfrak W\circ (Id\otimes E)$, where $E:\mathcal X\to\mathcal X_1\times\ldots\times\mathcal X_K$ is defined by $e(x_1,\ldots,x_K|x):=\prod_{i=1}^K\delta(x_i,x)$. Furthermore, it holds
	\begin{align}
		w_\mathrm{CI}(y_1,\ldots,y_K|x,s_1,\ldots,s_K)=\prod_{i=1}^Kw_i(y_i|x,s_i)
	\end{align}
	\end{definition}
	Due to the high mobility in vehicular connectivity, the antennas deployed on the bodywork may suffer under fast varying changes in channel states provoked by blocking of the electromagnetic waves and/or interference from other users competing for spectral resources. Though the antennas in a vehicular DAS are spaced sufficiently far apart to justify the assumption of independent channel conditions in terms of channel noise, the change in channel state caused by the fast changing blocking and interference effects occurs approximately simultaneously at all the MRFs. To account for this practical conditions, we limit the jammer in the AVC model by imposing additional constraints on its strategy. This results in a local state-constrained composite independent AVC. For this reason, we assume that the jammer can only use a single input jointly controlling all $K$ channels, meaning that $s_K=\ldots=s_1$ in Definition~\ref{def:Composite_orthogonal_AVC} and Definition~\ref{def:Composite_independent_AVC}. This is a realistic assumption for communication models where the probabilities in \eqref{eq:AVC_general} are derived from linear superposition of electromagnetic waves (in frequency or time) at the receive antennas and the Jammer is not able to adjust its channel input individually to the $K$ diversity branches.
	\begin{definition}[Local state-constrained composite independent AVC] Let $K\in\mathbb{N}$ and let $\mathcal{W}_1,\ldots,\mathcal{W}_K$ be AVCs with the input alphabet $\mathcal X$, state sets $\mathcal S_1,\ldots,\mathcal S_K$ and output alphabets $\mathcal Y_1,\ldots,\mathcal Y_K$. A local state-constrained composite independent AVC arising from these channels is defined via selecting a subset $\mathcal S\subset\mathcal S_1\times\ldots\times\mathcal S_K$ and setting
		\begin{align}
		\mathfrak W_{CI,lo-c}:=(W_{1,s_1}\otimes\ldots\otimes W_{K,s_K})_{s^K\in\mathcal S}.
		\end{align}
	A special case of the local state-constrained AVC is the identical state-constrained composite independent AVC $\mathfrak{W}_{CI,id-c}$ where $\mathcal S_1=\ldots=\mathcal S_K$ and $\mathcal S:=\{s^K:s_1=\ldots=s_K\}$. 
	\end{definition}
	\begin{remark}
		If an identical state-constrained jammer uses a probabilistic strategy where it selects its input at random according to a distribution $q\in\mathcal{P}(\mathcal{S})$, the effective distribution at the input of the identical state-constrained channel is 
		\begin{align}
		q^{(K)}:=\sum_{s\in\mathcal{S}}q(s)\delta_s^{\otimes K}.
		\end{align} 
		In addition, the set of all previously defined input distributions $q^{(K)}$ is defined in the following way:
		\begin{align}
		\mathcal{P}^{(K)}(\mathcal{S}):=\{q^{(K)}:q\in\mathcal{P}(\mathcal{S})\}.
		\end{align}
		Restricting the jammer to inputs of this form allows us to cast our analysis into the framework of dependent component analysis \cite{7572096} during the proof of Theorem~\ref{thm:sim_symm_BSC}.
	\end{remark} 
	\begin{remark}
	The identical state constraint can trivially be transferred to the composite orthogonal AVC. In this case, the identical state-constrained composite orthogonal AVC refers to the practical scenarios 1. and 2. mentioned at the end of the introduction. In contrast, the identical state-constrained composite independent AVC refers to the practical scenarios 3. and 4.
	\end{remark}
	Figure~\ref{fig:system_model} shows a block diagram of the previously discussed composite independent AVC under i.s.c. jamming. Following \cite{2627}, we use a power constraint in this contribution limiting the jammer's strategy in Theorem~\ref{thm:super_activ_sc_pw}.
	\begin{definition}[Jammer under state constraint]
	Let $K\in\mathbb N$ and $\mathcal{W}_1,\ldots,\mathcal{W}_K$ be AVCs with a local constraint set $\mathcal S$. Let $l:\mathcal S\to[0,\infty)$ and $\Lambda>0$. The set of states for a power constrained jammer is restricted to $\mathcal S_\Lambda^n:=\{s^n\in\mathcal S^n:\sum_{i=1}^nl(s_i)\leq\Lambda\cdot n\}$.
	\end{definition}
	Likewise, we limit the input strategy by the following definition.
	\begin{definition}[Input power constraint]
	Let $\mathcal{W}$ be an AVC with input alphabet $\mathcal X$. Let $g:\mathcal X\to[0,\infty)$ and $\Gamma>0$. The set of power-constrained input sequences is the restricted set $\mathcal X_\Gamma^n:=\{x^n\in\mathcal X^n:\sum_{i=1}^ng(x_i)\leq\Gamma\cdot n\}$.
	\end{definition}
	\begin{remark}
		We restrict ourselves to binary AVCs, that is, $\mathcal X=\mathcal Y=\{0,1\}$ with two states $\mathcal S=\{1,2\}$ and use $l(s):=s-1$ and $g(p)=p$ as cost functions. 
	\end{remark}
	\setlength\abovecaptionskip{-1pt}
	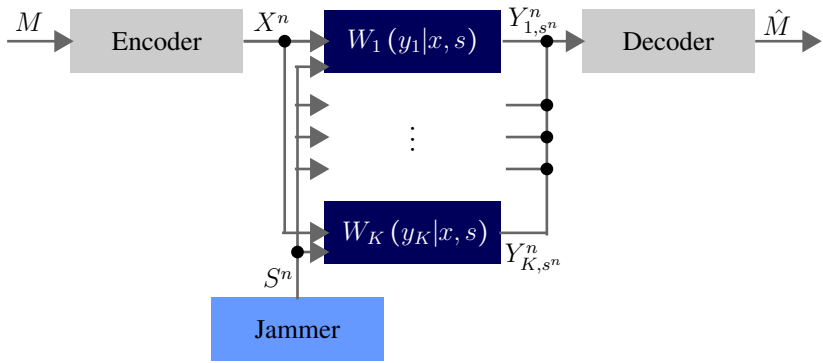
\begin{figure}\centering\begin{tikzpicture}[thick,scale=0.85, every node/.style={transform shape}]
		\node (T1)  [fill=light-gray,rectangle,inner sep=10pt, text width=2cm, align=center] at (0,3) {{Encoder}};
		\node (T2)  [fill=light-gray,rectangle,inner sep=10pt, text width=2cm, align=center] at (8,3) {{Decoder}};
		\node (T3)  [fill=darkblue,rectangle,inner sep=8pt, text width=2.2cm, align=center] at (4,3) {\textcolor{white}{$W_1\left(y_1\vert x,s\right)$}};
		\node (T4)  [fill=darkblue,rectangle,inner sep=8pt, text width=2.2cm, align=center] at (4,0) {\textcolor{white}{$W_K\left(y_K\vert x,s\right)$}};
		\node (T2)  
		[fill=lightblue,rectangle,inner sep=10pt, text width=2cm, align=center] at (2.2,-1.5) {{Jammer}};
		\path[color=dark-gray,line width=1pt, shorten <= -1pt] (2,3) edge (2,0);
		\path[-triangle 60, color=dark-gray,line width=1pt, shorten <= -1pt] (2,0) edge (2.7,0);
		\path[color=dark-gray,line width=1pt, shorten <= -1pt] (5.4,0) edge (6.1,0);
		\path[color=dark-gray,line width=1pt, shorten <= -1pt] (6.1,0) edge (6.1,3);
		\path[-triangle 60, color=dark-gray,line width=1pt, shorten <= -1pt] (1.4,3) edge (2.7,3);
		\path[-triangle 60, color=dark-gray,line width=1pt, shorten <= -1pt] (1.4,3) edge (2.7,3);
		\path[-triangle 60, color=dark-gray,line width=1pt, shorten <= -1pt] (5.45,3) edge (6.7,3);
		\path[-triangle 60, color=dark-gray,line width=1pt, shorten <= -1pt] (-2.3,3) edge (-1.3,3);
		\path[-triangle 60, color=dark-gray,line width=1pt, shorten <= -1pt] (9.4,3) edge (10.4,3);
		\path[color=dark-gray,line width=1pt, shorten <= -1pt] (2.2,-1) edge (2.2,2.6);
		\path[-triangle 60, color=dark-gray,line width=1pt, shorten <= -1pt] (2.2,2.6) edge (2.7,2.6);
		\path[-triangle 60, color=dark-gray,line width=1pt, shorten <= -1pt] (2.2,-0.3) edge (2.7,-0.3);
		\path[-triangle 60, color=dark-gray,line width=1pt, shorten <= -1pt] (2.2,2) edge (2.7,2);
		\path[-triangle 60, color=dark-gray,line width=1pt, shorten <= -1pt] (2.2,1.5) edge (2.7,1.5);
		\path[-triangle 60, color=dark-gray,line width=1pt, shorten <= -1pt] (2.2,1) edge (2.7,1);
		\path[color=dark-gray,line width=1pt, shorten <= -1pt] (5.5,1) edge (6.1,1);
		\path[color=dark-gray,line width=1pt, shorten <= -1pt] (5.5,1.5) edge (6.1,1.5);
		\path[color=dark-gray,line width=1pt, shorten <= -1pt] (5.5,2) edge (6.1,2);
		\node (Y1) [fill=black,ellipse,inner sep=2pt] at (2,3) {};
		\node (Y2) [fill=black,ellipse,inner sep=2pt] at (6.1,3) {};
		\node (Y3) [fill=black,ellipse,inner sep=2pt] at (2.2,-0.3) {};
		\node (Y4) [fill=black,ellipse,inner sep=2pt] at (6.1,1) {};
		\node (Y5) [fill=black,ellipse,inner sep=2pt] at (6.1,1.5) {};
		\node (Y6) [fill=black,ellipse,inner sep=2pt] at (6.1,2) {};
		\node (C1) at (1.9,-0.7)[text width=2cm,align=center] {\textcolor{black}{{$S^n$}}};
		\node (C2) at (1.8,3.3)[text width=2cm,align=center] {\textcolor{black}{{$X^n$}}};
		\node (C3) at (5.9,3.3)[text width=2cm,align=center] {\textcolor{black}{{$Y_{1,s^n}^n$}}};
		\node (C4) at (5.9,-0.4)[text width=2cm,align=center] {\textcolor{black}{{$Y_{K,s^n}^n$}}};
		\node (C4) at (-2,3.3)[text width=2cm,align=center] {\textcolor{black}{{$M$}}};
		\node (C4) at (9.7,3.3)[text width=2cm,align=center] {\textcolor{black}{{$\hat M$}}};
		\node[label=below:\rotatebox{-90}{\ldots}] at (4,2) {};
		\end{tikzpicture}
		\caption{Block diagram for an identical state-constrained composite independent AVC. The state of all $K$ parallel channels is controlled by the same state sequence $s^n\in\mathcal S^n$ chosen by the jammer.}
		\label{fig:system_model}
		\vspace{-10pt}
	\end{figure}
\end{section}
\begin{section}{Code Concepts and Performance Measures}\label{sec:code_concepts_and_perf_measures}
	In the following, we define code, rate and capacity for communicating over the i.s.c. composite independent AVC. The definitions can be generalized to the setting of i.s.c. composite orthogonal AVCs, replacing the input alphabet $\mathcal{X}$ by $\mathcal{X}_1\times\ldots\times\mathcal{X}_K$ allowing for more flexibility in precoding.
	\begin{definition}[Unassisted code]
	An unassisted (deterministic) code $\mathcal{M}_n$ for the i.s.c. composite independent AVC $\mathcal{W}\subset\mathcal{C}(\mathcal{X}\times\mathcal{S},\mathcal{Y}_1\times\ldots\times\mathcal{Y}_K)$ with $K\in\mathbb{N}$ consists of: a set $[M]$ of messages and a deterministic encoder $f:[M]\to\mathcal{X}^n$ in combination with a collection $\{\mathcal{D}_m\}_{m=1}^{M}$ of decoding subsets $\mathcal{D}_m\subset\mathcal{Y}_1\times\ldots\times\mathcal{Y}_K$ for which $\mathcal{D}_m\cap \mathcal{D}_{m'}=\emptyset$ for every $m\neq m'$. The average error of the code $\mathcal{M}_n$ is given by
	\begin{align*}
	\overline{e}_\mathrm{UA}(\mathcal{M}_n)=\max\limits_{s^n\in\mathcal{S}^n}\frac{1}{M}\sum\limits_{m=1}^{M}w^{\otimes n}(\mathcal{D}^c_m\vert f(m),s^n).
	\end{align*} 
	\end{definition}
		Advanced encoding and decoding schemes can be provoked by access to an additional coordination resource, that is, a random variable $\Gamma$ (common randomness) shared by the legitimate communication parties. Transmitter and Receiver can make use of $\Gamma$ to coordinate their choice of en- and decoders in order to ensure reliable communication avoiding symmetrizability.
	\begin{definition}[Common randomness-assisted code]\label{def:CR_ass_code}
		A common randomness-assisted code $\mathcal{M}_n$ for the i.s.c. composite independent AVC $\mathcal{W}\in\mathcal{C}(\mathcal{X}\times\mathcal{S},\mathcal{Y}_1\times\ldots\times\mathcal{Y}_K)$ with $K\in\mathbb{N}$ consists of: a set $[M]$ of messages, a set of random outcomes of a source of common randomness $[\Gamma]$ and a set of stochastic encoders $E^{\gamma}\in\mathcal{C}([M],\mathcal{X}^n)$ in combination with a collection $\{D_m^\gamma\}_{m,\gamma=1}^{M,\Gamma}$ of decoding subsets $D_m^\gamma$ for which $D_m^\gamma\cap D_{m'}^\gamma=\emptyset$ for all $\gamma\in[\Gamma]$, whenever $m\neq m'$. The average error of the CR-assisted code $\mathcal{M}_n$ is given by
		\begin{align}
			\overline{e}_{RA}(\mathcal{M}_n)=1-\min_{s^n\in\mathcal{S}^n}\frac{1}{M\cdot\Gamma}\sum\limits_{m,\gamma=1}^{M,\Gamma}e^\gamma(x^n\vert m)w_{s^n}(D_m^\gamma\vert x^n).
		\end{align}
	\end{definition}
	\begin{definition}[Achievable rate]
		A non-negative number $R$ is called an \textit{achievable rate} for the i.s.c. composite independent AVC $\mathcal{W}_{CI,id-c}\in\mathcal{C}(\mathcal{X}\times\mathcal{S},\mathcal{Y}_1\times\ldots\times\mathcal{Y}_K)$ with $K\in\mathbb{N}$ under average error criterion, if for every $\epsilon>0$ and $\delta>0$ and $n$ sufficiently large, there exists an unassisted code $\mathcal{M}_n$ such that $\frac{\log\vert \mathcal{M}_n\vert}{n}>R-\delta$, and $\overline{e}_{UA}<\epsilon$.\\
		Equally, an achievable rate $R$ for a CR-assisted code $\mathcal{M}_n$ for $\mathcal{W}_{CI,id-c}\in\mathcal{C}(\mathcal{X}\times\mathcal{S},\mathcal{Y}_1\times\ldots\times\mathcal{Y}_K)$ with $K\in\mathbb{N}$ under average error criterion is given, if for every $\epsilon>0$ and $\delta>0$ and $n$ sufficiently large, there exists a a CR-assisted code $\mathcal{M}_n$ such that $\frac{\log\vert \mathcal{M}_n\vert}{n}>R-\delta$, and $\overline{e}_{RA}<\epsilon$.
	\end{definition}
	\begin{definition}[Capacity]
		Let $K\in\mathbb{N}$. Given $K$ independent AVCs $\mathcal{W}_1,\ldots,\mathcal{W}_K\in\mathcal{C}(\mathcal{X}\times\mathcal{S},\mathcal{Y})$, the deterministic capacity of the composite independent AVC under identical input constrained jammer $\mathfrak W_{CI,id-c}$ is defined as
		\begin{align}
		C_d(\mathfrak W_{CI,id-c}):=\sup\left\{R:\begin{array}{l}R\mathrm{\ is\ an\ achievable\ rate\ for\ a\ deterministic\ coding}\\
		\mathrm{scheme\ under\ state\ constraint\ }\mathcal{S}=\{s^K:s_1=\ldots=s_K\}\end{array}\right\}.
		\end{align}
		The CR-assisted capacity of $\mathfrak W_{CI,id-c}$ is defined as
		\begin{align}
		C_r(\mathfrak W_{CI,id-c}):=\sup\left\{R:\begin{array}{l}R\mathrm{\ is\ an\ achievable\ rate\ for\ a\ CR-assisted\ coding}\\
		\mathrm{scheme\ under\ state\ constraint\ }\mathcal{S}=\{s^K:s_1=\ldots=s_K\}\end{array}\right\}.
		\end{align}
		Furthermore, we define the deterministic capacity of the composite independent AVC under state constraint $\Lambda$ and input power constraint $\Gamma$
		\begin{align}
		C_d(\mathfrak W_{CI,id-c},\Lambda,\Gamma):=\sup\left\{R:\begin{array}{l}R\mathrm{\ is\ an\ achievable\ rate\ for\ a\ deterministic\ coding}\\
		\mathrm{scheme\ under\ state\ constraint\ }\mathcal S=\{s^K:s_1=\ldots=s_K\}\\
		\mathrm{and\ power\ constraints}\ \Lambda\ and\ \Gamma.\end{array}\right\}.
		\end{align}
		Likewise, the CR-assisted capacity under power constraints is defined in the following way
		\begin{align}
		C_r(\mathfrak W_{CI,id-c},\Lambda,\Gamma):=\sup\left\{R:\begin{array}{l}R\mathrm{\ is\ an\ achievable\ rate\ for\ a\ common\ randomness}\\
		\mathrm{assisted\ coding\ scheme\ under\ state\ constraint\ }\\
		\mathcal S=\{s^K:s_1=\ldots=s_K\}\mathrm{\ and\ power\ constraints}\ \Lambda\ and\ \Gamma.\end{array}\right\}.
		\end{align}
	\end{definition}
	For our notion of symmetrizability, we stick to the standard definition in the context of arbitrarily varying channels which can, for example, be found in \cite{720535}. 
		\begin{definition}[Symmetrizability]\label{def:sim_symm}
		An AVC $\mathcal{W}\in\mathcal{C}(\mathcal{X}\times\mathcal{S},\mathcal{Y})$ is called symmetrizable, if for some $U\in\mathcal{C}(\mathcal{X},\mathcal{S})$,
		\begin{align}
		\sum\limits_{s\in\mathcal{S}}w(y\vert x,s)u(s\vert x')=	\sum\limits_{s\in\mathcal{S}}w(y\vert x',s)u(s\vert x),
		\end{align}
		for every $x,x'\in\mathcal{X}$, $y\in\mathcal{Y}$.
		\end{definition}
	In \cite{605589} and \cite{612938}, this definition was adapted to Arbitrarily Varying Multiple Access Channels (AVMAC) developing the notion of partial symmetrizability.\\
	For transferring the results concerning super-activation from the composite independent AVC to the composite orthogonal AVC, we make use of the following function, which was originally introduced in a quantum setting in \cite{doi:10.1063/1.4902930} and \cite{6874891} and later used to investigate classical communication channels in \cite{7447794}, to quantify the distance of an AVC from being symmetrizable.
	\begin{definition}\label{def:function_F} The function $F:\mathcal C(\mathcal{X}\times\mathcal{S},\mathcal{Y})\to\mathbb R_+$ is defined via setting
		\begin{align}\label{eqn:definition-of-F}
		F(\mathfrak W'):=\min_{U\in\mathcal{C}(\mathcal X,\mathcal S)}\max_{x\neq x'}\left\|\sum_{s\in\mathcal S}W'(\delta_{x'}\otimes\delta_s)u(s\vert x)-\sum_{s\in\mathcal S}W'(\delta_x\otimes\delta_s)u(s\vert x')\right\|_1,
		\end{align}
	for each $\mathfrak{W}'=(w'(\cdot\vert\cdot,s))_{s\in\mathcal S}\in \mathcal C(\mathcal{X}\times\mathcal{S},\mathcal{Y})$.
	\end{definition}
\begin{remark}
	Comparing Definition~\ref{def:sim_symm} and Definition~\ref{def:function_F}, it is obvious that $F(\mathcal{W}')=0$ is equivalent to $\mathcal{W}$ being symmetrizable.
\end{remark}
	Next, we provide a formal definition for the phenomenon of super-activation.
\begin{definition}[Super-activation]\label{def:super_act}
	Let the composite independent AVC $\mathfrak{W}_{CI,id-c}$ be given with $\mathcal{W}_1,\dots,\mathcal{W}_K\in\mathcal{C}(\mathcal{X}\times\mathcal{S},\mathcal{Y})$, $K\in\nn$. We say that the capacity $C$ of $\mathfrak W_{CI}$ shows super-activation if we can find $\mathcal{W}_1,\dots\mathcal{W}_K$ such that
	\begin{align}
	C(\mathcal{W}_1)=C(\mathcal{W}_2)=\dots=C(\mathcal{W}_K)=0,
	\end{align}  
	but 	
	\begin{align}
	C(\mathcal{W}_1\otimes\dots\otimes\mathcal{W}_K)>0.
	\end{align}
\end{definition}
	In the following, we stick to AVBSCs and show that the use of three independent AVBSCs, or alternatively, three orthogonal AVBSCs in a composite setting is already sufficient to circumvent symmetrizability except in some special, but trivial and practically irrelevant cases.
\end{section}
\begin{section}{Main Results}\label{sec:main_result}
	In the main part of this contribution, we precisely concentrate on the influence of receive diversity on the performance of the communication over AVCs. Our main results show why the deployment of a third antenna or, alternatively, the use of three orthogonal frequencies in a diversity fashion, can be an essential aspect for reliability in arbitrarily varying communication scenarios, especially if the communication channels are AVBSCs.
\begin{subsection}{Symmetrizability and Capacity of the i.s.c. Composite Independent AVBSC}\label{subsec:symm_and_cappa_res}
	\begin{theorem}[Non-symmetrizability for i.s.c. composite independent AVBSCs]\label{thm:sim_symm_BSC}
	Let $K=3$, $\mathcal X=\mathcal Y=\{0,1\}$ and $\mathcal S=[2]$. Let $\mathfrak{W}_\mathrm{CI,id-c}$ denote the i.s.c composite independent AVBSC consisting of $\mathcal{W}_1=\{W_{1,1},W_{1,2}\}$, $\mathcal{W}_2=\{W_{2,1},W_{2,2}\}$ and $\mathcal{W}_3=\{W_{3,1},W_{3,2}\}$. Associated to each channel state $W_{i,j}$ is the corresponding BSC-parameter $w_{i,j}$. Let, in addition, $w_{i,1}\neq w_{i,2}$ and $w_{i,j}\neq\nicefrac{1}{2}$ for all $i\in[3],j\in[2]$ and $w_{1,1}\neq1-w_{1,2}\lor w_{2,1}\neq1-w_{2,2}\lor w_{3,1}\neq1-w_{3,2}$. Then $\mathfrak{W}_\mathrm{CI,id-c}$ is not symmetrizable according to Definition~\ref{def:sim_symm}.
\end{theorem}
\begin{remark}
	Theorem~\ref{thm:sim_symm_BSC} implies that a third antenna can prevent the AVBSCs from symmetrizability and thus may evade zero capacity of the composite communication channel. At first, the exception seems to be a drawback, but on closer inspection, it has little practical relevance. In the exceptional case, the first channel state of each of the single AVBSCs is a flipped version of the second. This, however, seems rather unlikely to occur in practice. In addition, if a channel $BSC(p)$ has the property $p(1)=\nicefrac{1}{2}$, this channel completely randomizes every input. This is a situation that will typically trigger the strongest efforts towards improvement by other parts of the communication system. If $w_{i,1}= w_{i,2}$ for any $i\in[3]$, the jammer has no influence at all on the i-th communication link which can in this case be modeled by a simple DMC.
\end{remark}
Next, we have a look at the symmetrizability of the i.s.c. composite independent AVBSCs for $K=2$ and the symmetrizability of a single AVBSC ($K=1$).
\begin{theorem}[Symmetrizability of the composite independent AVBSC for $K=2$]\label{thm:symm_for_K=2}
	Let $\mathcal S=[2]$ and $K=2$. Let $\mathfrak{W}_{CI,id-c}$ denote the is.c. composite independent AVBSC consisting of $\mathcal{W}_1$ and $\mathcal{W}_2$ with states $W_{1,1},W_{1,2}$ and $W_{2,1},W_{2,2}$, respectively. Let $w_{2,2}=1-w_{2,1}$. Then $\mathfrak{W}_{CI,id-c}$ is symmetrizable according to Definition~\ref{def:sim_symm} for all tuples ($w_{1,1},w_{1,2}$) fulfilling $w_{1,2}=1-w_{1,1}$.
\end{theorem}
\begin{remark}
	In the special case where $w_{1,2}=1-w_{1,1}$, $K=2$ is already sufficient to prevent symmetrization except for the case in which it holds that $w_{2,2}=1-w_{2,1}$.
\end{remark}
\begin{theorem}[Symmetrizability of the AVBSC for $K=1$]\label{thm:symm_for_K=1}
	Let $\mathcal{S}=[2]$. Let $\mathcal{W}\in\mathcal{C}(\mathcal{X}\times\mathcal{S},\mathcal{Y})$ be an AVBSC with states $W_{1,1}$ and $W_{1,2}$. Let $w_{1,1}\in[0,\nicefrac{1}{2})$ be fixed. Then $\mathcal W$ is symmetrizable according to Definition~\ref{def:sim_symm} whenever $w_{1,2}\in(\nicefrac{1}{2},1]$, and the symmetrizing strategy may be chosen to be any channel $U\in \mathcal C(\mathcal X,\mathcal S)$ satisfying the equation
	\begin{align}
		u(2\vert 1)-u(1\vert 0)=(w_{1,1}+w_{1,2}-1)/(w_{1,2}-w_{1,1}).\label{eq:strategy_for_u}
	\end{align}
	 Note that it holds $(w_{1,1}+w_{1,2}-1)/(w_{1,2}-w_{1,1})\in(-1,1)$ for the given restrictions imposed on $w_{1,1}$ and $w_{1,2}$, so that many symmetrizers $U$ may exist. 
\end{theorem}
\begin{remark}
	Theorem~\ref{thm:symm_for_K=1} states that for every $w_{1,2}\in(\nicefrac{1}{2},1]$ there exists a jammer's strategy symmetrizing $\mathcal{W}$ resulting in zero channel capacity. Subsuming, the set of symmetrizable AVBSCs is large for $K=1$. This implies that a connectivity system relying solely on a single receive antenna operating in an AVC scenario is not able to guarantee a fail-safe operation over a wide range of system parameters. Coming back to the applicative example of vehicular communication for safety reasons, such a system would fail in any certification process if specific requirements on the reliability were imposed since overall system breakdown cannot be precluded. Comparing the result for $K=1$ to the conditions for symmetrizability of independent AVCs in Theorem~\ref{thm:sim_symm_BSC}, it is obvious that by using two additional, independent diversity branches for communicating over AVCs, it is possible to restrict the set of symmetrizable channels to a small subset of trivial channels with little importance for practical applications. 
\end{remark}
In the following example, we explicitly underline the difference in symmetrizability between $K=1$, $K=2$ and $K=3$, showing the shrinkage of the cardinality of the set of symmetrizable AVCs for the three cases comparing the results of the Theorems~\ref{thm:sim_symm_BSC} to \ref{thm:symm_for_K=1}.
\begin{example}
	First, we take the result of Theorem~\ref{thm:symm_for_K=1}. Recall that for a single AVC ($K=1$) and a fixed strategy $q^{(1)}$, every AVBSC $\mathcal{W}$ is symmetrizable for which for the BSC coefficients $w_{1,1}$ and $w_{1,2}$ of the respective states $W_{1,1}$ and $W_{1,2}$ it holds that, without loss of generality, $w_{1,1}\in[0,\nicefrac{1}{2})$ and $w_{1,2}\in(\nicefrac{1}{2},1]$.\\
	Next, let $\mathcal{U}_{CS}:=\{(w_{1,1},w_{1,2}):w_{1,1}\in[0,\nicefrac{1}{2}), w_{1,2}\in(\nicefrac{1}{2},1]\}$, and define two-state binary symmetric AVCs $\mathcal W(w_{1,1},w_{1,2}):=\{BSC(w_{1,1}),BSC(w_{1,2})\}$ for $(w_{1,1},w_{1,2})\in \mathcal{U}_{CS}$. Define the sets $\mathcal{U}_K((w_{2,1},w_{2,2})\dots (w_{K,1},w_{K,2}))$, $K\in\nn$, as follows:
	\begin{align}
	\begin{split}
	\mathcal{U}_K((w_{2,1},w_{2,2})\dots (w_{K,1},w_{K,2})):=\left\{\begin{array}{ll@{}}	
	(w_{1,1},w_{1,2})\in \mathcal{U}_{CS}:\exists (w_{2,1},w_{2,2})\dots (w_{K,1},w_{K,2}),\\
	\mathrm{such\ that\ }W(w_{1,1},w_{1,2})\otimes_{i=2}^K\mathcal W(w_{i,1},w_{i,2})\mathrm{\ is}\\
	\mathrm{symmetrizable\ for\ an\ i.s.c.\ jammer.}\end{array}\right\}.
	\end{split}
	\end{align}
	We know from Theorem~\ref{thm:symm_for_K=1} that $\mathcal{U}_1()=\mathcal{U}_{CS}$ for $K=1$. For $K=2$, Theorem~\ref{thm:symm_for_K=2} shows that $\mathcal{U}_2((w_{2,1},1-w_{2,1}))=\{(w_{1,1},w_{1,2})\in \mathcal{U}_{CS}:w_{1,2}=1-w_{1,1}\}$. In these two cases, $\mathcal{U}$ is a plane ($K=1$) or line with slope $-1$ ($K=2$), respectively. For a visualization, see (a) and (b) in Figure~\ref{fig:increasing_K_example}. Thus, the set of possible strategies causing symmetrization shrinks in the geometrical representation from a plane to a line, increasing $K$ from one to two.\\
	Next, we turn to the case in which three independent AVC diversity branches are available for communication, that is, $K=3$. This allows the use of the result of Theorem~\ref{thm:sim_symm_BSC} which asserts that for $K=3$ with $w_{3,2}\neq 1-w_{3,1}$, we have
	\begin{align}\label{eq:M_for_K_3}
	\mathcal{U}_3((w_{2,1},1-w_{2,1}),(w_{3,1},w_{3,2}))=\emptyset.
	\end{align}
	Thus, for $K=3$, $\mathcal{U}_3$ is the empty set. For a visualization, see (c) in Figure~\ref{fig:increasing_K_example}.\\
	\definecolor{darkblue}{RGB}{0,0,150}
		\setlength\abovecaptionskip{-1pt}
		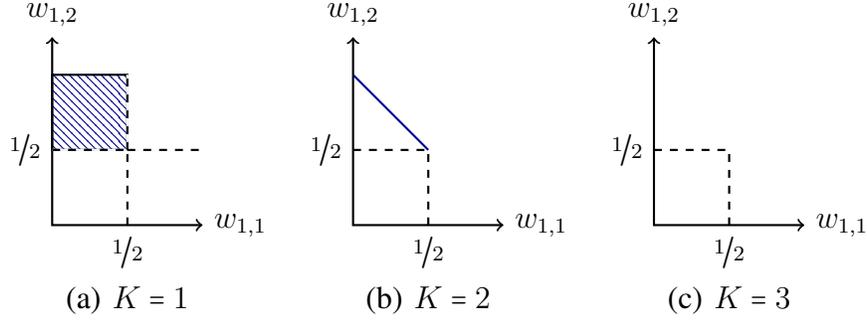
\begin{figure}\centering\begin{tikzpicture}[thick,scale=1, every node/.style={transform shape}]
		\draw [<->,thick] (0,2.5) node (yaxis) [above] {$w_{1,2}$}
		|- (2,0) node (xaxis) [right] {$w_{1,1}$};
		\draw [<->,thick] (4,2.5) node (yaxis2) [above] {$w_{1,2}$}
		|- (6,0) node (xaxis2) [right] {$w_{1,1}$};
		\draw [<->,thick] (8,2.5) node (yaxis3) [above] {$w_{1,2}$}
		|- (10,0) node (xaxis3) [right] {$w_{1,1}$};
		\coordinate (c) at (1,1);
		\coordinate (d) at (5,1);
		\coordinate (e) at (9,1);
		\draw[pattern=north west lines, pattern color=darkblue] (0,1) rectangle (1,2);
		\path[color=white] (1,1) edge (1,2);
		\path[color=white] (0,1) edge (1,1);
		\path[color=black, dashed] (1,1) edge (2,1);
		\path[color=black, dashed] (1,1) edge (1,2);
		\node (K1) at (1,-1)[align=center] {\textcolor{black}{(a)\ $K=1$}};
		\draw[dashed] (yaxis |- c) node[left] {$\nicefrac{1}{2}$}
		-| (xaxis -| c) node[below] {$\nicefrac{1}{2}$};
		\node (K2) at (5,-1)[align=center] {\textcolor{black}{(b)\ $K=2$}};
		\draw[dashed] (yaxis2 |- d) node[left] {$\nicefrac{1}{2}$}
		-| (xaxis2 -| d) node[below] {$\nicefrac{1}{2}$};
		\draw[darkblue] (4,2) --(d);
		\node (K3) at (9,-1)[align=center] {\textcolor{black}{(c)\ $K=3$}};
		\draw[dashed] (yaxis3 |- e) node[left] {$\nicefrac{1}{2}$}
		-| (xaxis3 -| e) node[below] {$\nicefrac{1}{2}$};
		\end{tikzpicture}	
		\caption{Symmetrizable i.s.c. composite independent AVBSCs with parameters $w_{1,1}$ and $w_{1,2}$. Depicted are the sets $\mathcal U_K$ for $K=1,2,3$.}
		\label{fig:increasing_K_example}
	\end{figure}
	The visualization of the shrinkage of the set $\mathcal{U}$ shows the importance of the concept of diversity for reliable communication. Communicating over an i.s.c. composite AVC consisting of three or more independent AVBSCs avoids symmetrizability, except in trivial cases. Enabling a fail-safe operation is one of the major concerns when standardizing and certificating concepts and techniques in the context of vehicular connectivity for safety reasons, or for highly-automated driving. Setting up a communication system according to the prerequisites of Theorem~\ref{thm:sim_symm_BSC}, it is possible to guarantee a positive capacity of a wireless communication link even in the presence of unknown and arbitrarily varying interference.
	\end{example}
	\begin{theorem}[AVBSC capacity]\label{thm:AVBSC_cappa}
		Let $K=3$, $\mathcal X=\mathcal Y=\{0,1\}$ and $\mathcal S=[2]$. Let the i.s.c. composite independent AVBSC $\mathfrak{W}_\mathrm{CI,id-c}$ with $K$ components be given such that the exception in Theorem~\ref{thm:sim_symm_BSC} is not fulfilled. Then the deterministic and the CR-assisted capacity, $C_\mathrm{d}$ and $C_\mathrm{r}$, of $\mathfrak{W}_\mathrm{CI,id-c}$ are given by
		\begin{align}\label{eq:id_cons_cappa}
		C_\mathrm{d}(\mathfrak W_\mathrm{CI,id-c})=C_\mathrm{r}(\mathfrak W_\mathrm{CI,id-c})=\min_{q\in\mathcal P(\mathcal{S})}\max_{p\in\mathcal{P}(\mathcal{X})} I\left(p;\sum_{s=1}^2q(s)\bigotimes_{i=1}^KW_{i,s}\right).
		\end{align}
	\end{theorem}
\begin{remark}
	Theorem~\ref{thm:AVBSC_cappa} together with the results from \cite{7447794} and \cite{doi:10.1063/1.4902930} using the distance measure defined in Definition~\ref{def:function_F} directly implies the continuity of $C_d$. Continuity implies that there exists an $\delta>0$ for all $\mathfrak{W'}\in\mathcal{C}([\mathcal{X}_1,\dots\mathcal{X}_K]\times\mathcal{S},\mathcal{Y}_1,\dots,\mathcal{Y}_K)$ that are 'close' in terms of the distance measure $d(\mathcal{W},\mathcal{\tilde{W}}):=\max_{x\in\mathcal{X}}\sum_{y\in\mathcal{Y}}|W(y|x)-,\mathcal{\tilde{W}}(y|x)|$. That is, $d(\mathfrak{W},\mathfrak{W'})<\delta$ such that $F(\mathfrak{W}')>0$ and thus $\mathfrak{W'}$ is not symmetrizable according to Definition~\ref{def:sim_symm}. This observation implies that continuity is guaranteed over a large parameter set.
\end{remark}
\end{subsection}
\begin{subsection}{Super-Activation of the Capacities for the i.s.c. Composite Independent AVBSC}\label{subsec:super_act_res}
	In \cite{7541865}, it is stated that '[...] super-activation is not possible for reliable message transmission over orthogonal AVCs' and thus, is a unique feature of orthogonal arbitrarily varying wiretap channels. For an unrestricted jammer being able to adapt its strategy on every channel separately, this statement is indisputable. However, the following theorem shows that for the i.s.c. composite independent AVBSC a completely unexpected phenomenon occurs.
\begin{theorem}[Super-activation of the deterministic and the CR-assisted capacity of the i.s.c. composite independent AVC]\label{thm:super_activ_sc}
	Let $K=3$, $\mathcal X=\mathcal Y=\{0,1\}$ and $\mathcal S=[2]$. Let the i.s.c. composite independent AVC $\mathfrak{W}_{CI,id-c}$, consisting of the three AVBSCs $\mathcal{W}_1=\{W_{1,1},W_{1,2}\}$, $\mathcal{W}_2=\{W_{2,1},W_{2,2}\}$ and $\mathcal{W}_3=\{W_{3,1},W_{3,2}\}$, be given. Let the BSC-parameters $w_{i,s}$ for the $i$-th channel satisfy $w_{1,1},w_{2,1},w_{3,1}\in[0,\nicefrac{1}{2})$ and $w_{1,2},w_{2,2},w_{3,2}\in(\nicefrac{1}{2},1]$. In addition, let $w_{1,1}\neq1-w_{1,2}\lor w_{2,1}\neq1-w_{2,2}\lor w_{3,1}\neq1-w_{3,2}$. Then it holds for $\mathcal{W}_1,\mathcal{W}_2,\mathcal{W}_3$:
	\begin{enumerate}
	 \item 	$C_d(\mathcal{W}_1)=C_d(\mathcal{W}_2)=C_d(\mathcal{W}_3)=0 \mathrm{\ but\ }C_d(\mathfrak{W}_{CI,id-c})>0$, that is, the deterministic capacity shows super-activation. 
	\item $C_r(\mathcal{W}_1)=C_r(\mathcal{W}_2)=C_r(\mathcal{W}_3)=0 \mathrm{\ but\ }C_r(\mathfrak{W}_{CI,id-c})>0$, that is, the CR-assisted capacity shows super-activation. 
	\end{enumerate}
	\end{theorem}
	\begin{remark}
		The second part of Theorem~\ref{thm:super_activ_sc} is, to the author's knowledge, the first occurrence of super-activation of the CR-assisted capacity. Hitherto, it was assumed that super-activation only occurs for deterministic capacities. Furthermore, symmetrizability was a decisive criterion for the deterministic capacity to be zero. For the CR-assisted capacity, symmetrizability is not as significant. Thus, part 2) of Theorem~\ref{thm:super_activ_sc} demonstrates an unexpected and surprisingly strong positive effect of receive diversity on the communication over AVCs.
	\end{remark}
	Below, we concentrate on the special case, $w_{1,1}=1-w_{1,2}\land w_{2,1}=1-w_{2,2}\land w_{3,1}=1-w_{3,2}$, excluded in Theorem~\ref{thm:super_activ_sc} for which we know by Theorem~\ref{thm:sim_symm_BSC} that the i.s.c. composite independent AVBSC is symmetrizable. In the following theorem, we show that also in this case, the deterministic capacity can be super-activated imposing power constraints on the transmitter's as well as the jammer's strategy. 
	\begin{theorem}[Super-activation of the deterministic capacity of the i.s.c. composite independent AVC with state and input power constraint]\label{thm:super_activ_sc_pw}
	Let $\mathcal S=[2]$, $K=3$ and $\mathcal X=\{0,1\}$. Let the i.s.c. composite independent AVC $\mathfrak{W}_{CI,id-c}$, consisting of the three AVBSCs $\mathcal{W}_1=\{W_{1,1},W_{1,2}\}$, $\mathcal{W}_2=\{W_{2,1},W_{2,2}\}$ and $\mathcal{W}_3=\{W_{3,1},W_{3,2}\}$, be given. Let the BSC-coefficients $w_{i,s}$ for the $i$-th channel satisfy $w_{1,1}=1-w_{1,2}$, $w_{2,1}=1-w_{2,2}$ and $w_{3,1}=1-w_{3,2}$ and all BSC parameters are not equal to $\nicefrac{1}{2}$. Let the state constraint $\Lambda=\nicefrac{1}{2}-\epsilon$ with cost function $l(s_i)=1-s$ and input-power constraint $\Gamma=\nicefrac{1}{2}-\eps$ with cost function $g(x_i)=x$ be given. There exists $\epsilon>0$ such that the deterministic capacity $C_d$ of $\mathfrak W_{CI,id-c}$ shows super-activation according to Definition~\ref{def:super_act}, that is,
	\begin{align}
	C_d(\mathcal{W}_1,\Lambda,\Gamma)=C_d(\mathcal{W}_2,\Lambda,\Gamma)=C_d(\mathcal{W}_3,\Lambda,\Gamma)=0 \mathrm{\ but\ }C_d(\mathfrak{W}_{CI,id-c},\Lambda,\Gamma)>0.
	\end{align}
\end{theorem}
\begin{theorem}[Super-activation of the capacity of the i.s.c. composite orthogonal AVBSC]\label{thm:super_activ_sc_pw_orth}
	The statements of Theorem~\ref{thm:super_activ_sc} and Theorem~\ref{thm:super_activ_sc_pw} remain valid when replacing the i.s.c. composite independent AVBSC with the i.s.c. composite orthogonal AVBSC.
\end{theorem}
\begin{remark}
	Theorem~\ref{thm:super_activ_sc}, Theorem~\ref{thm:super_activ_sc_pw} and Theorem~\ref{thm:super_activ_sc_pw_orth} demonstrate that super-activation is not only a phenomenon occurring under secrecy constraints. The results require state constraints which do not restrict the jammer in an unrealistic way concerning practical applications. Furthermore, limiting the jammer such that $\mathcal S=[2]$ has the practical implication of a two-level pulse amplitude modulation. Theorem~\ref{thm:super_activ_sc}, Theorem~\ref{thm:super_activ_sc_pw} and Theorem~\ref{thm:super_activ_sc_pw_orth} in combination with Theorem~\ref{thm:sim_symm_BSC} shed new light on the principle of symmetrizability and the methods and ways to avoid symmetrizability in real world applications by exploiting the concept of spatial and/or frequency diversity.
\end{remark}
	The identical state constraint imposed on the jammer's strategy allows for communicating at positive rates over composite AVCs. In contrast, it is obvious that by dispersing the constraint, individual symmetrizability of the orthogonal channels $\mathcal{W}_1,\dots,\mathcal{W_K}\in\mathcal{C}(\mathcal{X}\times\mathcal{S},\mathcal{Y})$ directly implies symmetrizability of the composite orthogonal AVC $\mathfrak{W}\in\mathcal{C}([\mathcal{X}_1\times\mathcal{X}_K]\times[\mathcal{S}_1\times\mathcal{S}_K],\mathcal{Y}_1\times\mathcal{Y}_K)$. In this case, the symmetrizing strategy is simply the product strategy. Thus, allowing the jammer to individually adjust its strategy on the respective sub-channels or frequency diversity branches, no reliable communication is possible in an uncoordinated setting, since symmetrizability may prohibit communicating over the composite AVC at positive rates. Subsuming, diversity in combination with a superiority of the legitimate communicating parties in terms of a higher number of antennas or more degrees of freedom concerning precoding compared to the jammer is an enabler for reliable communication over AVCs averting system breakdown.
\end{subsection}
\end{section}
\begin{section}{Proofs}\label{sec:proofs}
	\begin{subsection}{Proofs of the Results presented in Subsection~\ref{subsec:symm_and_cappa_res}}
	\begin{proof}[Proof of Theorem~\ref{thm:sim_symm_BSC}] Observe that for BSCs the following relation holds: For every $a,b\in\mathbb{R}$, we have $BSC(a)\circ BSC(b)=BSC(b)\circ BSC(a)$. To use the results presented in \cite{7572096}, we perform the following relabeling:
		\begin{align*}
		V_1(\delta_0)&=W_{1,1}(\delta_0)=\mathbf{w}_{1,1}, &V_1(\delta_1)&=W_{1,2}(\delta_0)=\mathbf{w}_{1,2},\\
		V_2(\delta_0)&=W_{2,1}(\delta_0)=\mathbf{w}_{2,1}, &V_2(\delta_1)&=W_{2,2}(\delta_0)=\mathbf{w}_{2,2},\\
		V_3(\delta_0)&=W_{3,1}(\delta_0)=\mathbf{w}_{3,1}, &V_3(\delta_1)&=W_{3,2}(\delta_0)=\mathbf{w}_{3,2},\\
		Y_1(\delta_0)&=W_{1,1}(\delta_1)=\mathbb{F}(\mathbf{w}_{1,1}), &Y_1(\delta_1)&=W_{1,2}(\delta_1)=\mathbb{F}(\mathbf{w}_{1,2}),\\
		Y_2(\delta_0)&=W_{2,1}(\delta_1)=\mathbb{F}(\mathbf{w}_{2,1}), &Y_2(\delta_1)&=W_{2,2}(\delta_1)=\mathbb{F}(\mathbf{w}_{2,2}),\\
		Y_3(\delta_0)&=W_{3,1}(\delta_1)=\mathbb{F}(\mathbf{w}_{3,1}), &Y_3(\delta_1)&=W_{3,2}(\delta_1)=\mathbb{F}(\mathbf{w}_{3,2}),
		\end{align*}
			where for a real number $x$, we define $\mathbf x:=(x,1-x)$. Observe that $Y_i(\delta_0)=\mathbb{F}\circ V_i(\delta_0)$ and $Y_i(\delta_1)=\mathbb{F}\circ V_i(\delta_1)$ for all $i\in[3]$. Now, the modified symmetrizability condition from Definition~\ref{def:sim_symm},
		\begin{align}\label{eq:symmetrizability_proof_1}
		\sum\limits_{s=1}^2 q(s)\bigotimes_{i=1}^3\left(W_{i,s}(\delta_0)\right)=\sum\limits_{s=1}^2 q'(s)\bigotimes_{i=1}^3\left(W_{i,s}(\delta_1)\right),
		\end{align}
		can equivalently be written in the following form:
		\begin{align}\label{eq:insert__reform_in_proof}
		\sum\limits_{s=1}^2q(s)\bigotimes_{i=1}^3\mathbf{w}_{i,s}=\sum\limits_{s=1}^2q'(s)(\mathbb{F}^{\otimes3})(\bigotimes_{i=1}^3\mathbf{w}_{i,s}),
		\end{align}
		where the ${w}_{j,i}$'s are vectors. Solving \eqref{eq:insert__reform_in_proof} for the jamming strategy $q^{(3)}:=\sum_{s=1}^2q(s)\delta_s^{\otimes 3}$, we obtain
		\begin{align}\label{eq:bring_it_to_DCA}
		q^{(3)}&=\left(\bigotimes_{i=1}^3(V_i^{-1}\circ\mathbb{F}\circ V_i)\right)q'^{(3)}=\left(\bigotimes\limits_{i=1}^3X_i\right)q'^{(3)},
		\end{align}
		where $X_i=V_i^{-1}\circ\mathbb{F}\circ V_i$ for all $i\in[3]$. The inverse of $V_i$ exists for all $i\in[3]$ because, by assumption, the AVC channel states are not equal. For binary alphabets, this translates to existence of the inverses. The special cases where $q(s)=q'(s')=1$ for some selection $s,s'\in\mathcal S$ lead to a very simple proof of the theorem: In order to see this, we explicitly compute the matrices $X_i$:
		\begin{align}\label{eqn:matrix_eq_1}
		X_i&=V_i^{-1}\circ\mathbb{F}\circ V_i\\
		&=\det(V_i)
		\begin{pmatrix}
		1-w_{i,2} &-w_{i,2}\\
		w_{i,1}-1 & w_{i,1}
		\end{pmatrix}
		\begin{pmatrix}
		0 & 1\\
		1 & 0
		\end{pmatrix}
		\begin{pmatrix}
		w_{i,1} & w_{i,2}\\
		1-w_{i,1} & 1-w_{i,2}
		\end{pmatrix}\\
		&=\frac{1}{w_{i,1}-w_{i,2}}
		\begin{pmatrix}
		1-w_{i,1}-w_{i,2} &1-2w_{i,2}\\
		-1+2w_{i,1} &-1+w_{i,1}+w_{i,2}
		\end{pmatrix}.\label{eqn:matrix_eq_2}
		\end{align}
		For the special cases mentioned above, equation \eqref{eqn:matrix_eq_1} has solutions $V_i$ with $i\in[3]$ which take either the form $V_i=BSC(\nicefrac{1}{2})$ or $V_i=BSC(w_{i,1})$. Since the first of this solutions is excluded by assumption, only the second case remains relevant and implies that in the above mentioned special cases, the theorem is proven. Next, we treat the more general case in which $q(1),q'(1)\in(0,1)$. Here, we can make us of \cite[Theorem~1]{7572096}: For every $\tau\in S_2$, it holds $X_1=X_2=X_3=\tau^{-1}$, $q'=\tau(q)$. It follows
		\begin{align}\label{eq:inverse_permutation}
		\frac{1}{w_{1,1}-w_{1,2}}
		\begin{pmatrix}
		1-w_{1,1}-w_{1,2} &1-2w_{1,2}\\
		-1+2w_{1,1} &-1+w_{1,1}+w_{1,2}
		\end{pmatrix}=\tau^{-1}.
		\end{align}
		The only permutation matrix for which \eqref{eq:inverse_permutation} holds is $\tau^{-1}=\mathbb{F}$. Thus all $V_i$'s are BSCs as a consequence of
		\begin{align}\label{eq:flip_relation}
		V_{i}=\mathbb F\circ V_{i}\circ \mathbb{F}\mathrm{\ for\ } i\in[3].
		\end{align}
		Combining \eqref{eq:flip_relation} and \eqref{eq:symmetrizability_proof_1} and using the fact that $q'(2)=q(1)$ (from \cite{7572096}), we obtain $W_{1,s}=\mathbb{F}\circ W_{1,s}$, $W_{2,s}=\mathbb{F}\circ W_{2,s}$ and $W_{3,s}=\mathbb{F}\circ W_{3,s}$, or alternatively, $q(1)=q(2)=\nicefrac{1}{2}$ and $w_{1,2}=\mathbb{F}(w_{1,1})$, $w_{2,2}=\mathbb{F}(w_{2,1})$ and $w_{3,2}=\mathbb{F}(w_{3,1})$. By assumption, this case is excluded. Thus, the theorem is proven.
	\end{proof}
	\begin{proof}[Proof of Theorem~\ref{thm:symm_for_K=2}]
		Recall that $w_{2,2}=1-w_{2,1}$. This implies that $V_2=BSC(w_{2,1})$. We again exploit the fact that for BSCs for every two BSC-coefficients, $a,b\in\mathbb R$, it holds that $BSC(a)\circ BSC(b)=BSC(b)\circ BSC(a)$. Applying this relation to \eqref{eqn:matrix_eq_1} for $i=2$ results in $V_2^{-1}\circ\mathbb{F}\circ V_2=V_2^{-1}\circ V_2\circ\mathbb{F}=\mathbb{F}$ and thus \eqref{eq:bring_it_to_DCA} reads as
		\begin{align}
			q^{(2)}=\left(X_1\otimes \mathbb{F}\right)q'^{(2)},\label{eq:bring_it_to_DCA_rm}
		\end{align}
		Next, we use the the partial trace $\tr_{[L']}:\mathbb R^L\otimes\mathbb R^{L'}\to\mathbb R^L$ summing over the 'content' of $\mathbb R^{L'}$ in the following way: For $v=\sum_{i,j=1}^{L,L'}v_{i,j}\delta_i\otimes \delta_j$, the trace operator is defined by
		\begin{align}
		\tr_{[L']}(v):=\sum_{i,j=1}^{L,L'}v_{i,j}\delta_j.
		\end{align}
		Tracing over the first system directly gives $q'=\mathbb{F}(q)$. Now, we go back to \eqref{eq:bring_it_to_DCA_rm} and make use of the following relation:
		\begin{align}\label{eq:X_for_K_2}
		q^{(2)}&=(X_1\otimes \mathbb{F})q'^{(2)}\\
		&=(X_1\otimes Id)(Id\otimes \mathbb{F})q'^{(2)}\\				
		&=(X_1\otimes Id)\sum_{j,i}q'(j)(1-\delta(i,j))\delta_i\otimes\delta_j\\
		&=(X_1\otimes Id)(q(1)\delta_2\otimes\delta_1+q(2)\delta_1\otimes\delta_2)\\
		&=\sum_{i,j=1}^{2}q(j)x_{(i,1\oplus j)}\delta_i\otimes\delta_j,
		\label{eq:eq:X_for_K_2_end}
		\end{align}
		where $\oplus:\mathbb{Z}\times(\mathbb{Z}\setminus\{0\})\to\mathbb{Z}, (a,b)\mapsto a\oplus b:=a+(b \mod 2)$. In the special cases where, for example $q(1)=q'(2)=1$ or $q(2)=q'(1)=1$, the solution set for equation \eqref{eq:X_for_K_2} is given by
		\begin{align}
		\left\{\begin{pmatrix}
		x & 1\\
		1-x & 0
		\end{pmatrix},
		\begin{pmatrix}
		0 & x\\
		1 & 1-x
		\end{pmatrix}\right\},
		\end{align}
		 which, in combination with \eqref{eqn:matrix_eq_1} and \eqref{eqn:matrix_eq_2}, lets us conclude that $w_{1,1}=1-w_{1,2}$ has to hold. In all other cases, component-wise comparison yields $X_1=X_2=\mathbb F$. Via equation \eqref{eq:inverse_permutation} this implies that $w_{1,1}=1-w_{1,2}$ holds in all these cases as well.
		\end{proof}
	\begin{remark}
		The general case for $K=2$ remains an open problem, since the $X_i$'s are no BSCs in general. If all $X_i$ were BSCs by assumption, the problem could be solved by using Theorem~3 in \cite{noetzel2016}.
	\end{remark}
	\begin{proof}[Proof of Theorem~\ref{thm:symm_for_K=1}]
	Let the transition probability matrices of the two AVBSCs with states $W_{1,1}$ and $W_{1,2}$ be defined as follows:
	\begin{align}
	W_{1,1}&=\begin{pmatrix}
	w_{1,1} &1-w_{1,1}\\
	1-w_{1,1} &w_{1,1}
	\end{pmatrix},
	&W_{1,2}&=\begin{pmatrix}
	w_{1,2} &1-w_{1,2}\\
	1-w_{1,2} &w_{1,2}
	\end{pmatrix}.
	\end{align}
	Recall that we excluded the trivial case in which $w_{1,1}=w_{1,2}=\nicefrac{1}{2}$. Now, we check the criterion for symmetrizability:
	\begin{align}
	&u(1\vert 0)(w_{1,1}\cdot\delta_0+(1-w_{1,1})\delta_1)+(1-u(1\vert 0))(w_{1,2}\cdot\delta_0+(1-w_{1,2})\delta_1)\\
	&=u(1\vert 1)((1-w_{1,1})\delta_0+w_{1,1}\cdot\delta_1)+(1-u(1\vert 1))((1-w_{1,2})\delta_0+w_{1,2}\cdot\delta_1).	
	\end{align}
	This is equivalent to the following two equations
	\begin{align}
	u(1\vert 0)\cdot w_{1,1}+(1-u(1\vert 0))w_{1,2}&=u(1\vert 1)(1-w_{1,1})+(1-u(1\vert 1))(1-w_{1,2}),\label{eqn_ex_avc_norm_1}\\
	u(1\vert 0)(1-w_{1,1})+(1-u(1\vert 0))(1-w_{1,2})&=u(1\vert 1)\cdot w_{1,1}+(1-u(1\vert 1))w_{1,2}.\label{eqn_ex_avc_norm_2}
	\end{align}
	Since $w_{1,1}\in[0,\nicefrac{1}{2})$ is fixed, it is obvious that the above equations, \eqref{eqn_ex_avc_norm_1} and \eqref{eqn_ex_avc_norm_2}, can be fulfilled for arbitrary $w_{1,2}\in(\nicefrac{1}{2},1]$ by choosing $u(1\vert 0)$ and $u(1\vert 1)$ properly.
	\end{proof}
	\begin{proof}[Proof of Theorem~\ref{thm:AVBSC_cappa}]
	Theorem~\ref{thm:AVBSC_cappa} follows from Theorem~1 in \cite{2627} and the observation that Theorem~\ref{thm:sim_symm_BSC} ensures non-symmetrizability. Thus, the deterministic capacity equals its common randomness-assisted counterpart.
	\end{proof}
\end{subsection}
\begin{subsection}{Proofs of the Results presented in Subsection~\ref{subsec:super_act_res}}
	\begin{proof}[Proof of Theorem~\ref{thm:super_activ_sc}]
		The first part of 1) in Theorem~\ref{thm:super_activ_sc} is given by Theorem~\ref{thm:symm_for_K=1}. For the second part of 1), Theorem~1 in \cite{Ahlswede1978} in combination with \cite{2627} and \cite{1056995} implies that the deterministic code capacity of a non-symmetrizable AVC under average probability of error either equals its CR-assisted capacity or else is zero. Furthermore, we know by Theorem~\ref{thm:sim_symm_BSC} that the i.s.c. composite AVBSC $\mathfrak W_{CI,id-c}$ is non-symmetrizable. Thus, it holds that $C_d(\mathfrak W_{CI,id-c})=C_r(\mathfrak W_{CI,id-c})$. By the same argumentation as in \cite{2627}, we show that non-symmetrizability of $\mathfrak W_{CI,id-c}$ is a sufficient condition for $C_r(\mathfrak W_{CI,id-c})>0$: Assume for contradiction that it holds 
		\begin{align}\label{eqn:mutual_inf_zero}
			\min_{q\in\mathcal P(\mathcal{S})}\max_{p\in\mathcal{P}(\mathcal{X})} I\left(p;\sum_{s=1}^2q(s)\bigotimes_{i=1}^KW_{i,s}\right)=0
		\end{align}
		for the non-symmetrizable i.s.c. composite independent AVBSC $\mathfrak W_{CI,id-c}$. This implies that the outputs of the AVBSCs $\mathcal{W}_1,\mathcal{W}_2$ and $\mathcal{W}_3$ of $\mathfrak W_{CI,id-c}$ in \eqref{eqn:mutual_inf_zero} do not depend on $p(x)$ for $p(x)> 0$ for all $x\in\mathcal{X}$, that is,
		\begin{align}
			P_{Y_1,Y_2,Y_3}(y_1,y_2,y_3)=\sum\limits_{s\in\mathcal{S}}W_1(y_1\vert x,s)W_2(y_2\vert x,s)W_3(y_3\vert x,s)q(s)
		\end{align} 
		does not depend on $x$. This results in constant independent channels. This in turn implies trivial symmetrizability by setting $u(\cdot\vert x)=q(\cdot)$ which contradicts the assumption. Thus, non-symmetrizability of $\mathfrak W_{CI,id-c}$ implies $C_r(\mathfrak W_{CI,id-c})>0$. This proves 1) of Theorem~\ref{thm:super_activ_sc}.\\
		For the proof of 2) in Theorem~\ref{thm:super_activ_sc}, the remaining part is to show that $C_r(\mathcal{W}_1)=C_r(\mathcal{W}_2)=C_r(\mathcal{W}_3)=0$. For this reason, we define the convex hull of a set of AVBSCs as
		\begin{align}
			\conv(\mathcal{W}):=\left\{W=\sum_{s=1}^{2}q(s)W_s:q\in\mathcal{P}(\mathcal{S})\right\}.
		\end{align}
		We know by \cite{blackwell1960,AHLSWEDE1969457,stiglitz1966coding} that the CR-assisted capacity of a single AVC is given by
		\begin{align}
			C_r(\mathcal{W})=\min_{q\in\mathcal{P}(\mathcal{S})}\max_{p\in\mathcal{P}(\mathcal{X})}I\left(p,\sum_{s=1}^{2}q(s)W_s\right).\label{eq:single_avc_cappa}
		\end{align}
		Observe that for each of the individual AVBSCs $\mathcal{W}_i$ with $i\in[3]$ it holds that $BSC(\nicefrac{1}{2})\in\conv(\mathcal{W}_i)$ for which $I(p,BSC(\nicefrac{1}{2}))=0$. Thus, by \eqref{eq:single_avc_cappa} and the fact that minimization and maximization can be interchanged in \eqref{eq:single_avc_cappa}, which is a consequence of the convexity of the mutual information with respect to the channel, it holds that $C_r(\mathcal{W}_i)=0$ for all $i\in[3]$.
	\end{proof}
	\begin{proof}[Proof of Theorem~\ref{thm:super_activ_sc_pw}]
	The proof of Theorem~\ref{thm:super_activ_sc_pw} makes use of the results published in \cite{2627}. First, we show that for the deterministic capacity of the individual channels it holds $C_\mathrm{d}(\mathcal{W}_i,\Lambda,\Gamma)=0$ for all $i\in[3]$: For $K=1$, without loss of generality, $\mathcal{W}_1$, $\mathcal{W}_2$ and $\mathcal{W}_3$ are symmetrizable for $w_{1,1},w_{2,1},w_{3,1}\in[0,\nicefrac{1}{2})$ and $w_{1,2},w_{2,2},w_{3,2}\in(\nicefrac{1}{2},1]$ when considered separately. In order to derive the implications of the power constraint on the symmetrizing strategy, we concentrate on $\mathcal{W}_1$ in the following. The symmetrizing strategy for $K=1$ reads as follows:
	\begin{align}
	u(2\vert 1)-u(1\vert 0)=\frac{w_{1,1}+w_{1,2}-1}{w_{1,2}-w_{1,1}}\label{eqn:symmetrizing_strategy}.
	\end{align}
	Let $\mathcal{U}$ denote the set of channels $\mathcal{C}(\mathcal{S},\mathcal{X})$ fulfilling symmetrizability according to Definition~\ref{def:sim_symm}. Observe that for $w_{1,2}=1-w_{1,1}$, \eqref{eqn:symmetrizing_strategy} delivers that the AVBSCs $\mathcal{W}_1$ is symmetrizable for every symmetrizer $U\in\mathcal C(\mathcal X,\mathcal S)$ fulfilling $u(2\vert 1)-u(1\vert 0)=0$. Thus, $U$ is itself a BSC and, moreover, $Id$ and $\mathbb F$ both fall into the class of symmetrizers. Let $\mathcal U$ be the set of channels symmetrizing $\mathcal W_1$. Now, we refer to \cite{2627} and compute 
	\begin{align}
	\Lambda_0(p)&=\min\limits_{U\in\mathcal{U}}\sum_{x\in\mathcal{X}}^{}\sum_{s\in\mathcal{S}}^{}p(x)u(s\vert x)l(s)\\
	&=\min\limits_{U\in\mathcal{U}}(p\cdot u(2\vert 1)+(1-p)\cdot u(2\vert 2))\\
	&=\min\limits_{u\in[0,1]}(p(1-2u)+u)\label{eq:last_step_p_u}\\
	&=\begin{cases}
	1-p & \text{if $p\geq \frac{1}{2}$} \\
	p & \text{otherwise}
	\end{cases},
	\end{align}
	where \eqref{eq:last_step_p_u} follows from the fact that $U$ is a BSC with BSC-parameter $u$. The last step is due to the maximum for $u$ being attained at the boundaries of the interval $[0,1]$.	Observe that $\Lambda_0(p,1-p)\in[0,\nicefrac{1}{2}]$ where $\Lambda_0(p,1-p)=\nicefrac{1}{2}$ is attained for $p=\nicefrac{1}{2}$. Now, we make use of Theorem~3 in \cite{2627}. Likewise, we define $g(p)=\sum_{x\in\mathcal{X}}^{}p(x)g(x)$. Observe that for our setup for the individual AVBSCs it holds
	\begin{align}
	\max_{g(p)\leq\Gamma}\Lambda_0(p,1-p)=\nicefrac{1}{2}-\eps.
	\end{align}
	Then, Theorem~3 in combination with \cite[Remark on page 188]{2627} gives $C_\mathrm{d}(\mathcal{W}_1,\Lambda,\Gamma)=0$. Especially, one may choose $\eps\in(0,\nicefrac{1}{2})$ such that $U_1,U_2,U_3$ obey the power constraint $\Lambda=\nicefrac{1}{2}-\epsilon$ for $\mathcal W_1,\mathcal W_2,\mathcal W_3$ leading to $C_\mathrm{d}(\mathcal{W}_i,\Lambda,\Gamma)=0$ for all $i\in[3]$. Observe that for the i.s.c. composite independent AVBSC $W_\mathrm{CI,id-c}$ and $\eps\in(0,\nicefrac{1}{2})$ it holds that
	\begin{align}
	\max_{g(p)\leq\Gamma}\Lambda_0(p,1-p)=\nicefrac{1}{2}>\nicefrac{1}{2}-\eps=\Lambda,
	\end{align}
	since the minimization over $U\in\mathcal{U}$ is carried out over the set of symmetrizing strategies which for the i.s.c. composite independent AVBSC is the one element set $\mathcal{U}=\{\nicefrac{1}{2}\}$. Thus, by part 2) of Theorem~3 in \cite{2627}, we get $C_\mathrm{d}(\mathfrak{W},\Lambda,\Gamma)>0$.
\end{proof}
	\begin{proof}[Proof of Theorem~\ref{thm:super_activ_sc_pw_orth}]
	We know from the proof of Theorem~3 in \cite{7447794} that positivity of the function $F(\mathcal{W})$ from Definition~\ref{def:function_F} implies that $\mathcal{W}'$ is non-symmetrizable. The existence of the effect of super-activation of the capacity of the i.s.c. composite orthogonal AVBSC is then a direct consequence of Theorem~\ref{thm:sim_symm_BSC} in combination with the observation that by allowing different inputs for the respective diversity branches (orthogonal AVBSCs), the maximization in \eqref{eqn:definition-of-F} is carried out over a larger input-alphabet set. The preceding considerations hold for both, the deterministic as well as the CR-assisted capacity. Considering all $K$ channels separately, there is no difference between the composite independent and the composite orthogonal AVBSC, that is, every AVBSC obeying the conditions stated in the theorem is symmetrizable since $BSC(\nicefrac{1}{2})\in\conv(\mathcal W_i)$ for all $i\in[K]$. This results in zero deterministic as well as CR-assisted capacity. In contrast, when turning over to the joint use of the $K$ orthogonal channels, Theorem~\ref{thm:sim_symm_BSC} ensure non-symmetrizability and thus positivity for both capacities.
	\end{proof}
\end{subsection}
\end{section}
\begin{section}{Discussion and Conclusion}\label{sec:discussion_and_conclusion}
In this contribution, we investigated the impact of receive diversity on reliable communication. Our model advances on previous studies by taking into account the potential impact of a jammer. We proved that receive diversity has the potential to dramatically increase both the CR-assisted and the non-assisted capacity. Most importantly, it turned out that there exist many channels which may independently be symmetrizable, but become non-symmetrizable when operated in SIMO mode when the jammer cannot adapt to the channels individually. Such communication systems can especially be found in V2X communication. Precisely speaking, we showed that for composite independent AVCs under identical state constraints, the deployment of three uncorrelated receive antennas or the communication over three orthogonal frequencies already avoids symmetrizability and thus zero capacity, if all the channels under consideration are independent or orthogonal AVBSCs. We explicitly highlighted that the deterministic capacities of composite independent and orthogonal AVBSCs under an i.s.c. jammer are continuous and show the phenomenon of super-activation if an additional power constraint is imposed on the jammer's and the transmitter's strategy or if channel states of little practical relevance are excluded. For the second restriction, we proved that also the CR-assisted capacity shows super-activation. The appearance of super-activation in the considered setting is the strongest form of violation of the additivity statement for orthogonal AVCs, since for the general unconstrained composite orthogonal binary symmetric AVC super-activation does not occur, but super-additivity does. The capacity of unrestricted orthogonal AVCs is upper-bounded by the sum of the CR-assisted capacities. Thus, the CR-assisted capacity can be seen as the benchmark quantity for characterizing the performance of communication over an AVC. In this contribution, we showed that by using diversity not only stabilizing the communication over AVCs, but even lifting the CR-assisted capacity away from zero is possible. This recognition is of particular importance for control channels in modern communication systems for which high reliability has to be ensured. This is because failures in the control channel directly lead to failures in the whole communication system provoking system breakdown.\\
Summarizing, we presented the concept of receive diversity as an enabler for super-activation and, thus, reliable communication over a channel with arbitrarily varying interference. The results, presented in this contribution, can be transferred to arbitrary scenarios in which multiple receivers are cooperating in AVCs with an i.s.c. jammer. In order to benefit from the introduced concept, one major prerequisite is that cooperation is done lossless and the channels of the transmitter to the cooperating receivers are independent. The latter assumption is achieved by spatial separation of the cooperating parties or the usage of multiple, orthogonal frequencies. An important question for further analysis is how this result can be adapted to more general AVCs of arbitrary alphabet size. However, this seems to be a rather complicated problem since it is extremely hard to analyze \eqref{eqn:matrix_eq_2} when the alphabets are non-binary. In addition, it is of great interest to find a concrete implementation of a cooperation protocol. The concept of Willems conferencing used in \cite{6384744}, in the context of AVMACs with conferencing encoders, may be a promising approach. 
\end{section}
\begin{acknowledgement}
	C.A. thanks P. Fertl and A. Posselt from the BMW Group for supporting his studies. Funding is acknowledged from DFG via grant BO 1734/20-1, BMBF via grants 01BQ1050 and 16KIS0118 (H.B.), DFG via grant NO 1129/1-1, BMWi and ESF via grant 03EFHSN102, the ERC Advanced Grant IRQUAT, the Spanish MINECO Project No. FIS2013-40627-P and the Generalitat de Catalunya CIRIT Project No. 2014 SGR 966 (J.N.).
\end{acknowledgement}
\pagebreak

\end{document}